%% file: 0_main-file.tex
\newif\ifSyntaxChecking

\newif\ifeurocg
\eurocgfalse

\newif\ifLipics
\Lipicsfalse

\newif\ifIntThreeCon
\IntThreeContrue

\newif\ifFalse
\Falsetrue

\documentclass[a4paper,11pt,twoside]{article}
\usepackage[margin=1in]{geometry}

\usepackage{fancyhdr}
\newcommand{\proofof}[1] {Proof of #1}

\usepackage[labelfont=bf]{caption}

\usepackage[utf8]{inputenc}
\usepackage{subcaption}
\usepackage{wrapfig}

\usepackage{xcolor}

\usepackage{xspace}
\usepackage{enumerate}

\usepackage{microtype}

\usepackage{amsmath}
\usepackage{amsthm}
\usepackage{amssymb}
\usepackage{amstext}
\usepackage{esvect}
\usepackage{amsopn}
\usepackage{soul}

\usepackage{hyperref}
\usepackage{graphicx}
\usepackage[ruled,vlined]{algorithm2e}
\SetCommentSty{textrm}
\SetKwComment{tcp}{$\triangleright$ }{}
\SetKwFor{ForAll}{for all}{do}{}
\usepackage[basic]{complexity}

\usepackage{paralist}
\usepackage{enumitem}
\usepackage{thm-restate}

\ifSyntaxChecking
\input prepareForSyntaxChecking

\else
\usepackage[left, pagewise, displaymath, mathlines]{lineno}
\newcommand*\patchAmsMathEnvironmentForLineno[1]{%
  \expandafter\let\csname old#1\expandafter\endcsname\csname #1\endcsname
  \expandafter\let\csname oldend#1\expandafter\endcsname\csname end#1\endcsname
  \renewenvironment{#1}%
     {\linenomath\csname old#1\endcsname}%
     {\csname oldend#1\endcsname\endlinenomath}}%
\newcommand*\patchBothAmsMathEnvironmentsForLineno[1]{%
  \patchAmsMathEnvironmentForLineno{#1}%
  \patchAmsMathEnvironmentForLineno{#1*}}%
\AtBeginDocument{%
\patchBothAmsMathEnvironmentsForLineno{equation}%
\patchBothAmsMathEnvironmentsForLineno{align}%
\patchBothAmsMathEnvironmentsForLineno{flalign}%
\patchBothAmsMathEnvironmentsForLineno{alignat}%
\patchBothAmsMathEnvironmentsForLineno{gather}%
\patchBothAmsMathEnvironmentsForLineno{multline}%
}
\fi

\usepackage[textsize=footnotesize]{todonotes}
\SetVlineSkip{0pt}  
\def\comment#1{}%
\def\withcomments{%
  \newcounter{mycommentcounter}%
   \def\comment##1{\refstepcounter{mycommentcounter}%
    \ifhmode%
     \unskip%
     {\dimen1=\baselineskip \divide\dimen1 by 2 %
       \raise\dimen1\llap{\tiny
	{-\themycommentcounter-}}}\fi%
     \marginpar[{\renewcommand{\baselinestretch}{0.8}%
       \hspace*{-2em}\begin{minipage}{1.5\marginparwidth}\footnotesize%
[\themycommentcounter]:%
\raggedright ##1\end{minipage}}]{\renewcommand{\baselinestretch}{0.8}%
       \begin{minipage}{1.5\marginparwidth}\footnotesize%
[\themycommentcounter]: \raggedright%
##1\end{minipage}}}%
  }

\newcommand{\remove}[1]{{}}
\newcommand{\changed}[1]{#1}
\newcommand{\changednew}[1]{{#1}}


\newtheorem{theorem}{Theorem}
\newtheorem{lemma}[theorem]{Lemma}
\newtheorem{corollary}[theorem]{Corollary}
\newtheorem{obs}[theorem]{Observation}

\theoremstyle{remark}

\title{Convexity-Increasing Morphs of Planar Graphs\footnote{A preliminary version of this paper appeared in the proceedings of WG 2018~\cite{Kleist-WG-2018}.}} 
\author{
\hspace{6em}Linda Kleist\footnote{Technische Universit\"at Braunschweig, Germany}
\and
 Boris Klemz\footnote{Institut f\"ur Informatik, Freie Universit\"at Berlin, Germany}
  \and
Anna Lubiw\footnote{University of Waterloo, Canada}\hspace{6em}
\and
\hspace{5em}Lena Schlipf\footnote{FernUniversit\"at in Hagen, Germany}
\and
Frank Staals\footnote{Utrecht University, The Netherlands}
\and
Darren Strash\footnote{Hamilton College, USA}\hspace{5em}
} 
\date{}

\begin{document}

\maketitle
\thispagestyle{empty}

 \begin{abstract}
We study the problem of \emph{convexifying} drawings of planar graphs.
Given any planar straight-line drawing of an internally 3-connected graph, we show how to 
morph the drawing to one with strictly convex faces while maintaining planarity at all times. 
Our morph is \emph{convexity-increasing}, meaning that 
once an angle is convex, it remains convex.
We give an efficient algorithm that constructs such a
morph as a composition of a linear number of steps where each step either moves vertices along
horizontal lines or moves vertices along vertical lines. Moreover, we show that a linear number of steps is worst-case optimal.

 To obtain our result, we use a well-known technique by Hong and Nagamochi for finding redrawings with convex faces while preserving $y$-coordinates.
Using a variant of Tutte's graph drawing algorithm, we obtain a new proof of Hong and Nagamochi's result which comes with a better running time.
 This is of independent interest, as Hong and Nagamochi's technique serves as a building block in existing morphing algorithms.
\end{abstract}

\input{Introduction-new}

\input{Preliminaries}

\input{Algorithm}

\input{OuterFace}

\input{int3con}

\input{Tutte}

\input{spiral}

\input{GridSize}

\section{Conclusions}\label{sec:openProblems}

We have shown how to morph any straight-line planar drawing of an internally 3-connected graph to a strictly convex drawing while preserving planarity and increasing convexity throughout the morph. Moreover, our morph is composed of a linear number of horizontal and vertical steps, which is asymptotically optimal.
The following questions are open:

\begin{enumerate}
\item 
Recall that during a convexity-increasing morph, the set of \emph{internal} convex angles never decreases. 
We conjecture that every straight-line planar drawing of a (internally) 3-connected graph admits a convexity-increasing morph to a strictly convex drawing such that during the morph the set of \emph{external} reflex angles also never decreases.

\item Our algorithm for finding a convexity-increasing morph to a convex drawing can be executed in $O(n^{1 + \omega/2})$ time on a Real-RAM. 
Our new version of Lemma~\ref{lem:H&N} also speeds up the run-time of the algorithm of Alamdari et al.~\cite{alamdari2016morph} for morphing between two given drawings from $O(n^3)$ to $O(n^{1 + \omega/2})$. For both these problems, it would be interesting to find even more efficient algorithms,  or to establish non-trivial lower bounds on the run-time.

\item 
A main open question is to design piece-wise linear morphs with a polynomial bound on the bit complexity of the intermediate drawings.  
This would be a step towards having 
intermediate drawings that lie on a polynomial-sized grid, i.e.~with a logarithmic number of bits for each vertex's coordinates.
This is open both for our problem of morphing to a convex drawing and for the problem of morphing between two given planar straight-line drawings.

\item We have introduced the idea of using horizontal and vertical morphs.  
It would be interesting to further explore their visual quality or to explore what can be accomplished with this restriction. 
Horizontal and vertical morphs were very recently used by Da Lozzo et al.~\cite{da-lozzo2018upward} for a special case of morphing between two upward planar drawings.

\end{enumerate}

\medskip\noindent{\bf{Acknowledgments.}}
We 
thank Andr\'e Schulz for helpful discussions on generalizations of Tutte's algorithm.
This work was begun at Dagstuhl workshop 17072, ``Applications of Topology to the Analysis of 1-Dimensional Objects.''  We thank Dagstuhl, the organizers, and the other participants for a stimulating workshop.
In particular, we 
thank Carola Wenk and Regina Rotmann for joining some of our discussions, and Irina Kostitsyna for contributing 
many valuable ideas.

\bibliographystyle{abbrv} 
\bibliography{convexify}

\input{Appendix}

\end{document}

%% file: Introduction-new.tex
\section{Introduction}

 A \emph{morph} between two planar straight-line drawings $\Gamma_0$ and $\Gamma_1$ of a plane graph $G$ is a continuous movement of the vertices from their positions in~$\Gamma_0$ to their positions in~$\Gamma_1$, 
with the edges following along as straight-line segments between their endpoints. 
A morph is \emph{planar} if it preserves planarity of the drawing at all times.

Motivated by applications in animation and in reconstruction of 3D shapes from 2D slices, the study of morphing has focused on finding a morph between two given planar drawings.
The existence of planar morphs was established long ago~\cite{Cairns,Thomassen}, followed by algorithms that produce good visual results~\cite{Floater-Gotsman,Gotsman-S}, and algorithms that find ``piece-wise linear'' morphs with a  linear number of steps~\cite{alamdari2016morph}. 
Our focus is somewhat different, and  
more aligned with graph drawing goals---our input is a planar graph drawing and our aim is to morph it to a better drawing, in particular to a convex drawing. 
A morph \emph{convexifies} a given straight-line graph drawing if the result is a \emph{(strictly) convex graph drawing}, i.e.~a planar straight-line graph drawing in which 
every face is a (strictly) convex polygon.
\ifeurocg \else For a survey on convex graph drawing, see~\cite{Rahman2015}.\fi{}

We first observe that it is easy, using known results, to find a planar morph that convexifies a given drawing---we can just create a convex drawing with the same faces (assuming such a drawing exists), and morph to that specific drawing 
using the known planar morphing algorithms.
(For a discussion of the techniques used, see the section on related work.)
In this paper, we are interested in a stronger condition: we want to find a convexifying morph which is also 
\emph{convexity-increasing},
meaning that once an angle of an inner face is convex, it remains convex. 
We illustrate a convexity-increasing morph in Fig.~\ref{fig:exMorphSeq}.

\begin{figure}[thb]
\centering
\includegraphics{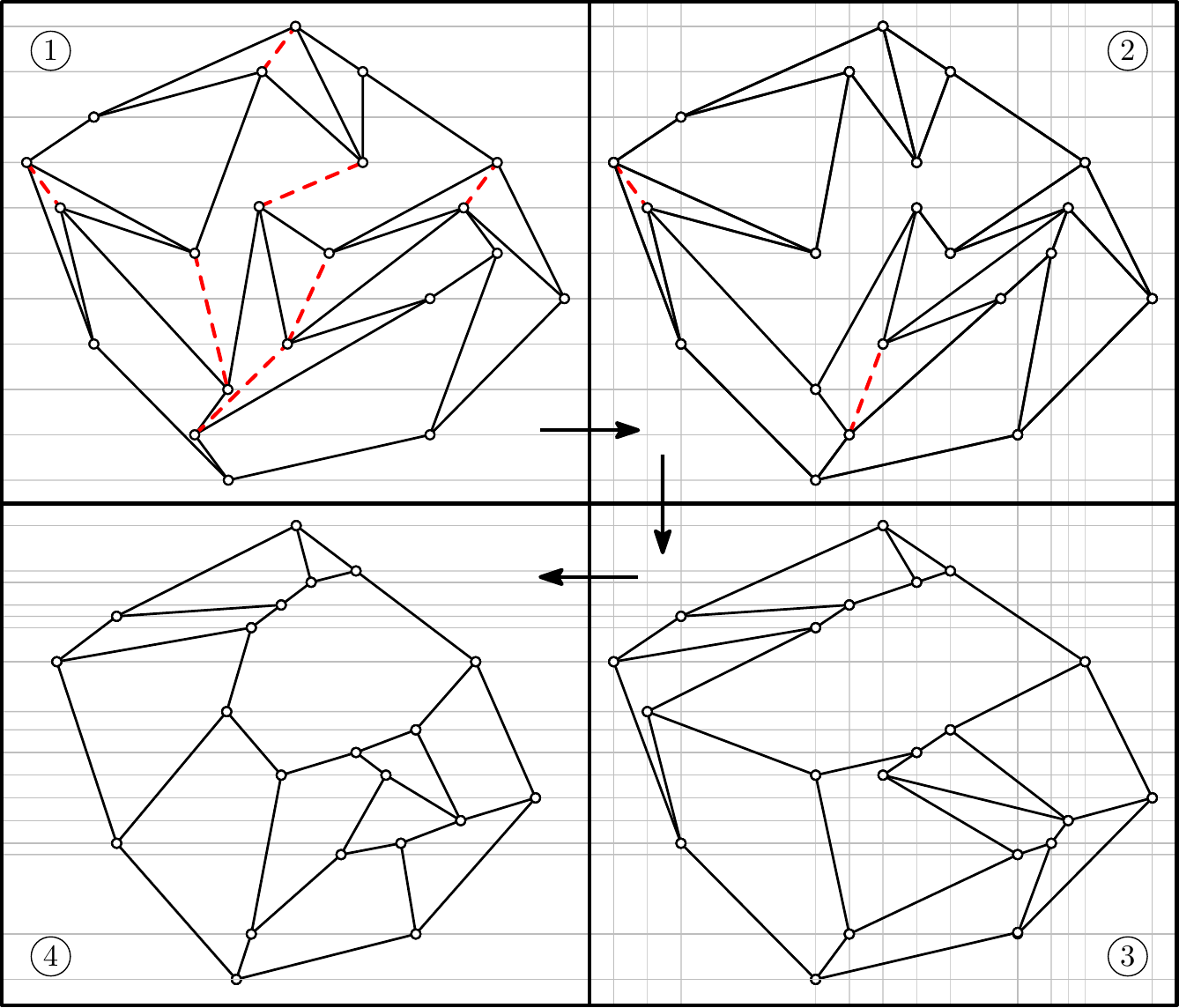}
\caption{A sequence of convexity-increasing morphs (horizontal, vertical,
  horizontal) that morph a straight-line drawing of a graph $G$ (1)
  into a strictly convex drawing of $G$ (4).
  The dashed segments are auxilliary edges added by our algorithm before each morphing step.}
\label{fig:exMorphSeq}
\end{figure}

Besides the theoretical goal of studying continuous motion that is monotonic in some measure (e.g.~edge lengths~\cite{iben2009refolding}), another motivation comes from 
visualization---a morph of a graph drawing should maintain the user's ``mental map''~\cite{purchase2006important} which means changing as little as possible, and making observable progress towards a goal. 
Most previous morphing algorithms fail to provide convexity-increasing morphs even if the target is a convex drawing because they 
start by triangulating the drawing.  Therefore, an original convex angle may be subdivided by new triangulation edges, so there is no constraint that keeps it convex.

 \smallskip\noindent{\bf{Related work.}}
To the best of our knowledge,  previous work on convexity-increasing morphs only considers the case when the input graph is a cycle (or a path).
Connelly et al.~\cite{Connelly} and Canterella et al.~\cite{Cantarella-Demaine} gave algorithms to convexify a simple polygon 
while preserving edge lengths.  Since their motions are ``expansive'', they are convexity-increasing.  
Aichholzer et al.~\cite{aichholzer2011convexifying} gave an algorithm to find a ``visibility-increasing'' morph of a simple polygon to a convex polygon; this condition 
implies the condition of being convexity-increasing.

In related work, there is an algorithm to morph a convex drawing to another convex drawing of the same graph while preserving planarity and convexity~\cite{angelini-convex-drawings-2015}.  Such morphs are convexity-increasing by default, but do not address our problem since our initial drawing is not convex.

\changed{Many previous morphing algorithms find}
``piece-wise linear'' morphs, where the morph is composed of discrete steps and each step moves vertices along straight lines.  
A morph is called \emph{linear}
if each vertex moves along a straight line at \changed{constant speed;
different vertices are allowed to move at different speeds, and some may remain stationary.}
A linear morph is completely specified by the initial and final drawings.  
If, in addition, all the lines along which vertices move are parallel, then the morph is called \emph{unidirectional}~\cite{alamdari2016morph}.  

Alamdari et al.~\cite{alamdari2016morph} gave an algorithm with runtime $O(n^3)$ that takes as input two planar straight-line drawings of a graph on $n$ vertices with the same combinatorial embedding, i.e., the drawings  have the same outer face and for each face (every boundary component) has the same cyclic ordering of edges. It then constructs a planar morph between the two drawings that consists of a sequence of $O(n)$ 
unidirectional morphs.

\smallskip\noindent{\bf{Our contribution.}} In this paper, we give the first algorithm that convexifies a given straight-line planar drawing~$\Gamma$ via a planar convexity-increasing morph. The only requirement is that the plane graph~$G$ represented by~$\Gamma$ admits a strictly convex drawing. This is the case if and only if~$G$ is internally 3-connected; see Section~\ref{sec:convex-exists} for the definition and related discussions.

In fact, we achieve the following stronger property---our morphs are 
 composed of a linear number of \emph{horizontal} and \emph{vertical} morphs.   A \emph{horizontal} morph moves all vertices at constant speeds along horizontal lines; a
\emph{vertical} morph is defined analogously. 
These are special cases of unidirectional morphs. See Fig.~\ref{fig:exMorphSeq} for an illustration.

Orthogonality is a very desirable and well-studied criterion for graph drawing~\cite{eiglsperger2001orthogonal}, in part because there is evidence that the human visual cortex comprehends orthogonal lines more easily~\cite{appelle1972perception,marriott2012memorability,purchase2012graph}. 
Similarly, it seems natural that orthogonal \emph{motion} should be easier to comprehend, though this criterion has not been explored in previous work.

Our main result is summarized in the following theorem.

\begin{theorem}
\label{thm:morph-to-convex}
Let $\Gamma$ be a planar straight-line drawing of an internally 3-connected graph $G$ on $n$ vertices.  Then $\Gamma$ can be morphed to a strictly 
convex drawing via a sequence of at most $3.5n+2$ convexity-increasing planar morphing steps
each of which is either  horizontal or vertical.  

In the special cases that $G$ is 3-connected or $\Gamma$ has a convex outer face, the upper bound on the  number of morphing steps can be improved to $1.5n+2$ or $\mathrm{max}\lbrace 2,r+1\rbrace$, respectively, where $r$ denotes the number of internal reflex angles.

Furthermore, there is an 
$O(n^{1 + \omega/2})$
time algorithm to find the sequence of morphs,
where $\omega$ is the matrix multiplication exponent.
\end{theorem}

The run time is $O(n^{2.5})$ with Gaussian elimination, improved to $O(n^{2.1865})$
using the current fastest matrix multiplication method with $\omega \approx 2.3728639$~\cite{LeGall}.
Our model of computation is the real-RAM---we do not have a polynomial bound on the bit-complexity of the coordinates of the vertices in the sequence of drawings that specify the morph.
However, previous morphing algorithms had no such bounds either.

Our algorithm has another advantage in terms of visualization over previous morphing algorithms such as the one by Alamdari et al.~\cite{alamdari2016morph}.  These algorithms tend to   
``almost contract'' vertices, which destroys the user's ``mental map'' of the graph. 
We do not use
contractions, 
and therefore
expect our morphs to 
be useful for visualizations.

A main ingredient of our proof
is a result of Hong and Nagamochi~\cite{hn-2012} that gives conditions (and an algorithm) for redrawing a planar straight-line drawing to have convex faces, while preserving the $y$-coordinates of the vertices (``level planar drawings of hierarchical-st plane graphs,'' in their terminology).  
Angelini et al.~\cite{angelini-convex-drawings-2015} strengthened Hong and Nagamochi's result to strictly convex faces. 
We give a new proof of the strengthened result using Tutte's graph drawing algorithm.
Thereby, we improve the runtime of  Hong and Nagamochi's result from $O(n^2)$ to  $O(n^{\omega/2})$.
This is of independent interest, as Hong and Nagamochi's technique serves as a building block in several
other morphing algorithms~\cite{alamdari2016morph,angelini-convex-drawings-2015,da-lozzo2018upward}.
In particular, our improvement also speeds up the run-time of the morphing algorithm of Alamdari et al.~\cite{alamdari2016morph} from $O(n^3)$ to $O(n^{1 + \omega/2})$:

\begin{theorem} [Theorem 1.1 in~\cite{alamdari2016morph} with an improved runtime]
\label{thm:improved-morph}
Given a planar graph $G$ on $n$ vertices and two straight-line planar drawings of
$G$ with the same combinatorial embedding, there is a planar morph between the two drawings that consists of $O(n)$ unidirectional morphs. Furthermore,
the morph can be found in time $O(n^{1 + \omega/2})$.
\end{theorem}

In Appendix~\ref{appendix:previousMorphing} we describe the algorithm by Alamdari et al.~and justify the improved runtime due to our version of Hong and Nagamochi's result.

Theorem~\ref{thm:morph-to-convex} guarantees the existence of a convexity-increasing morph to a strictly convex drawing where the morph is composed of~$O(n)$ horizontal/vertical morphs.
This is optimal in the worst case.
In fact we show something stronger:

\begin{restatable}{theorem}{spiral}
\label{thm:spiral}
For any $n\geq 3$, there exists a drawing of an internally $3$-connected graph on $n$ vertices for which any convexifying planar morph composed of a sequence of linear morphing steps requires $\Omega (n)$ steps.
\end{restatable}

\smallskip\noindent{\bf{Organization.}} 
Our paper is structured as follows: We begin with preliminaries in Section~\ref{sec:preliminaries}.
The proof of Theorem~\ref{thm:morph-to-convex} is presented in Section~\ref{sec:main-proof}, and the proof of the improved running time of Theorem~\ref{thm:improved-morph} is given in Section~\ref{sec:Tutte}.
The lower bound on the number of morphs, namely Theorem~\ref{thm:spiral}, is shown in Section~\ref{sec:spiral}.
Finally, a discussion of the size of the grid needed for the 
intermediate drawings of our morph can be found in 
Section~\ref{sec:grid}. We conclude with open problems in Section~\ref{sec:openProblems}.

%% file: Preliminaries.tex
\section{Preliminaries}
\label{sec:preliminaries}

In this section, we introduce the concepts we will use.  
We formally define convex drawings and internally 3-connected graphs in Section~\ref{sec:convex-exists} as well as $y$-monotone drawings in Section~\ref{sec:yMonotone}.
We proceed by stating several useful properties of unidirectional morphs in Section~\ref{sec:unidirectional}. 
Finally, we address the concept of finding convex drawings in Section~\ref{sec:H&N}.

\subsection{Convex Drawings and Internal $3$-Connectivity}
\label{sec:convex-exists}

Given a planar straight-line drawing $\Gamma$ of a graph, its \emph{angles} are formed by pairs of consecutive edges around a face, with the angle measured inside the face. 
An \emph{internal angle} is an angle of an inner face. 
\changed{We say an angle is \emph{reflex} it it exceeds~$\pi$, \emph{convex} if it is at most $\pi$, and  \emph{strictly convex} if it is less than $\pi$.}
A drawing $\Gamma$ is \emph{convex} if the boundary of every face is a convex polygon, i.e., angles of the inner faces are convex and angles of the outer face are reflex or of size $\pi$.
The drawing is \emph{strictly} convex if the boundary of every face is a strictly convex polygon. 

\smallskip\noindent{\bf{Conditions for the existence of convex drawings.}} 
Throughout, we assume that our input is a drawing of a graph that admits a strictly convex drawing  with the same combinatorial embedding.
Necessary and sufficient conditions for the existence of a strictly convex drawing were given by Tutte~\cite{tutte-1960}, Thomassen~\cite{Thomassen-1984}, and Hong and Nagamochi~\cite{hn-2012}. 
These conditions can be tested in linear time by the algorithm of Chiba et al.~\cite{chiba1985}.

Such conditions are usually stated for a fixed convex drawing of the outer face, but the conditions become simpler when, as in our case, the drawing of the outer face may be freely chosen---in particular, may be chosen to have no 3 consecutive collinear vertices.
Internal vertices of degree $2$ can also be dealt with directly:  In a convex drawing,
an internal vertex of degree $2$ must be drawn as a point in the interior of the
straight line segment formed by its two incident edges.
This has two implications.
Firstly, a graph with an
internal vertex of degree 2 has no strictly convex drawing.  Secondly, for a convex drawing we may eliminate every internal degree 2 vertex by repeatedly replacing a path
of two edges by a single edge.  However, if this produces multiple edges, then there
exists no convex drawing.

With these observations, the necessary and sufficient conditions for the existence
of a (strictly) convex drawing become quite simple to state.
A plane graph~$G$ is \emph{internally 3-connected} if the graph is 2-connected and
any separation pair $\lbrace u,v\rbrace$ is \emph{external}, meaning that $u$ and
$v$ lie on the outer face and that every connected component of $(G- u-v)$
contains a vertex of the outer face of~$G$.
Observe that the two neighbours of an internal vertex of degree 2 form a separation
pair that is not external. 
The results of Tutte~\cite{tutte-1960}, Thomassen~\cite{Thomassen-1984}, and Hong and Nagamochi~\cite{hn-2012}
 become:  

\begin{lemma}\label{lem:convConn}
Let $G$ be a plane graph with outer face $C$.  
Then 
\begin{compactenum}
\item $G$ has a strictly convex drawing with outer face $C$ if and only if $G$ is
internally 3-connected. 
\item $G$ has a convex drawing with outer face $C$ if and only if repeatedly
eliminating internal vertices of degree 2 produces a graph that has no multiple
edges and is internally 3-connected.
\end{compactenum}
 \end{lemma}

Note that a separation pair which is not external \emph{can} have both of its vertices on the outer face, see Fig.~\ref{fig:externalSepPairs}(b,c).
For this reason, we refer to a separation pair which is not external as \emph{non-external}\footnote{instead of using the more canonical, but misleading term \emph{internal}}.

\begin{figure}[bht]
  \centering
  \includegraphics{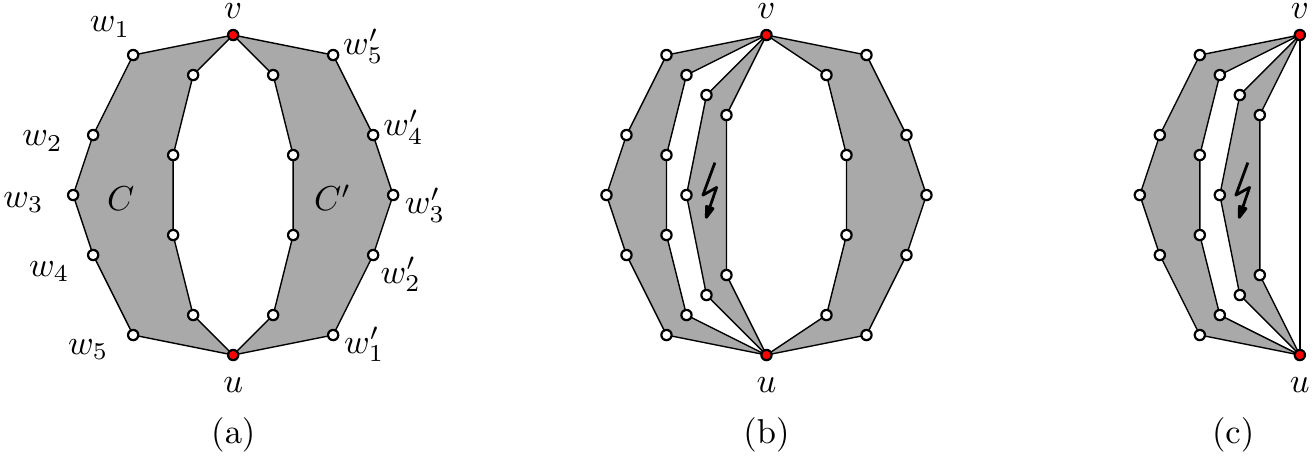}
  \caption{
  (a) An external separation pair $\lbrace u,v\rbrace$ and its two components $C$ and $C'$.  (b--c) In these cases $\lbrace u,v\rbrace$ is a non-external separation pair because the component marked with a jagged arrow has no vertex on the outer face.  In these cases there is no  convex drawing of $G$.
}
  \label{fig:externalSepPairs}
 \end{figure}

\smallskip\noindent{\bf{Structure of internally $3$-connected graphs.}}
There are multiple well-known equivalent definitions of internal $3$-connectivity.
Each of them provides a different perspective on the concept and it will be convenient to be able to refer to all of them.
Hence, we state the following characterization:

\begin{lemma}\label{lem:int3conDefs}
Let~$G$ be a plane $2$-connected graph and let~$f_o$ denote its outer face. The following statements are equivalent:
\begin{enumerate}[leftmargin=*,label={(I\arabic*)}]
\item \label{I1} $G$ is internally $3$-connected.
\item \label{I2} Inserting a new vertex~$v$ in~$f_o$ and adding edges between~$v$  and all vertices of~$f_o$ results in a $3$-connected graph.
\item \label{I3} From each internal vertex~$w$ of $G$ there exist three paths to~$f_o$ that are pairwisely disjoint except for the common vertex~$w$.
\end{enumerate}
\end{lemma}

\begin{proof}
\ref{I1} obviously implies \ref{I2}, which in turn implies \ref{I3} by  Menger's theorem.
It remains to show that \ref{I3} implies \ref{I1}.
So suppose that from each internal vertex~$w$ there exist three paths to~$f_o$ that are disjoint except for the common vertex~$w$.
It follows that for each pair of vertices~$p,q$ each connected component of~$(G-p-q)$ has at least one vertex on~$f_o$, as otherwise there can be at most two disjoint paths from~$w$ to~$f_o$.
Moreover, both~$p$ and~$q$ belong to~$f_o$:
Assume otherwise.
Since~$G$ is $2$-connected,~$f_o$ is a \emph{simple} cycle and, hence, all vertices of~$(f_o-p-q)$ belong to the same connected component of~$(G-p-q)$.
However, by assumption from each internal vertex~$w$ there exist at least one path to~$f_o$ in~$(G-p-q)$ and, so,~$(G-p-q)$ is connected; a contradiction.
\end{proof}

The following statement gives a characterization of external separation pairs and describes the structure of internally 3-connected graphs. For an illustration see Fig.~\ref{fig:externalSepPairs}(a).

\begin{obs}\label{lem:externalSepPair}
Let $H$ be a plane 2-connected graph and let~$\lbrace u,v\rbrace$ be a separation pair of $H$. Then, $\lbrace u,v\rbrace$ is external if and only if all of the following conditions hold:
\begin{enumerate}[leftmargin=*,label={(E\arabic*)}]
\item \label{E1}Vertices $u$ and $v$ belong to the outer face of $H$.
\item \label{E2}The outer face of~$H$ decomposes into two internally disjoint paths $(u, w_1, \dots, w_j, v)$ and $(v, w'_1,\dots ,w'_\ell, u)$ each with at least 3 vertices, i.e.~$j\ge 1$and~$l\ge 1$.
\item \label{E3}Vertices $w_1,\dots,w_j$ belong to a connected component $C$ of $(H-u-v)$.
\item \label{E4}Vertices $w'_1,\dots,w'_j$ belong to a connected component $C'$ of $(H-u-v)$.
\item \label{E5}The graph $(H-u-v)$ has no connected component other than $C,C'$.
\item \label{E6}The components $C$ and $C'$ are distinct.
\end{enumerate}
\end{obs}

\begin{proof}
If the six conditions hold, then clearly $\lbrace u,v\rbrace$ is an external separation pair.
On other hand, if $\lbrace u,v\rbrace$ is an external separation pair, $u$ and $v$ belong to the outer face of $H$ \ref{E1}.
Since~$H$ is $2$-connected, its outer face is a \emph{simple} cycle.
Further, the removal of~$u$ and $v$ splits the graph into at least two connected components each of which has a vertex that belongs to the outer face of $H$.
Hence, the removal of~$u$ and~$v$ decomposes the outer face into two internally disjoint paths $(u, w_1, \dots, w_j, v)$ and $(v, w'_1,\dots ,w'_\ell, u)$ each with at least 3 vertices \ref{E2}.
Since $(w_1, \dots, w_j)$ is a path in $(H-u-v)$, its vertices belong to a connected component $C$ \ref{E3}.
Similarly, $(w'_1, \dots, w'_j)$ is a path in $(H-u-v)$, its vertices belong to a connected component $C'$ \ref{E4}.
Since these two paths together with $u$ and $v$ cover the entire outer face of $H$, there can not be any more components \ref{E5}.
Finally, since there are at least two components, $C$ and $C'$ have to be distinct \ref{E6}.
\end{proof}

\subsection{$y$-Monotone Drawings}
\label{sec:yMonotone}

A face of a planar graph drawing is \emph{$y$-monotone} if the boundary of the face
consists of two $y$-monotone paths.
A path is \emph{$y$-monotone} if the $y$-coordinates along the curve realizing the path are strictly increasing.
These definitions apply to general planar graph drawings, not just straight-line drawings.
We note that the directed graphs that have drawings with $y$-monotone faces are the 
\emph{st-planar} graphs, which are well-studied~\cite{DiBattista-st-planar}.

We say a vertex $v$ is a \emph{local minimum} (\emph{local maximum}) of face $f$ in a drawing $\Gamma$ if the neighbors of  $v$ in $f$ lie above $v$ (below $v$, respectively).  
A \emph{local extremum} refers to a local minimum or a local maximum.
Note that a face $f$ is $y$-monotone if and only if it has exactly one local maximum and exactly one local minimum. Alternatively, a face is $y$-monotone if it has no reflex local extremum.

\subsection{Linear and Unidirectional Morphs}
\label{sec:unidirectional}

A linear morph is completely specified by the initial and the final drawing.  To denote the linear morph from a drawing $\Gamma_1$ to a drawing $\Gamma_2$, we use the notation $\langle  \Gamma_1,  \Gamma_2 \rangle$.
Restricting to linear morphs is a sensible way to discretize morphs---essentially, it asks for the vertex trajectories to be piece-wise linear.  At first glance, the restriction to unidirectional morphs  seems arbitrary and restrictive.  However, as discovered by Alamdari et al.~\cite{alamdari2016morph}, it  is easier to prove the existence of unidirectional morphs.  Also, unidirectional morphs have many nice properties,
as we explain in this section. 
Suppose we do a horizontal morph. 
Then every vertex keeps its $y$-coordinate.   Alamdari et al.~\cite{alamdari2016morph} gave conditions on the initial and final drawing that guarantee that the horizontal morph between them is planar:

\begin{lemma}~\cite[in the proof of Lemma 13]{alamdari2016morph}
If $\Gamma$ and $\Gamma'$ are two planar straight-line drawings of the same graph such that every line parallel to the $x$-axis crosses the same ordered sequence of edges and vertices in both drawings, then the linear morph from $\Gamma$ to $\Gamma'$ is planar.
\label{lem:uni-morph} 
\end{lemma}

Observe that the conditions of the lemma imply that every vertex is at the same $y$-coordinate in $\Gamma$ and $\Gamma'$ so the linear morph between them is horizontal.
Also note that the lemma generalizes in the obvious way to any direction, not just the direction of the $x$-axis.
We note several useful consequences of Lemma~\ref{lem:uni-morph}.

\begin{lemma}\label{lem:union}
 Let $\Gamma_1,\Gamma_2,\Gamma_3$ be three planar straight-line drawings where the linear morphs $\langle \Gamma_1,\Gamma_2 \rangle$ and $\langle \Gamma_2,\Gamma_3 \rangle$  are horizontal and planar.  
Then the linear morph $\langle \Gamma_1, \Gamma_3 \rangle$ is a horizontal planar morph.
\end{lemma}

\begin{proof} 
The morphs $\langle \Gamma_1, \Gamma_2 \rangle$ and $\langle \Gamma_2, \Gamma_3 \rangle$ are horizontal and planar, so every line parallel to the $x$-axis crosses the same ordered sequence of edges and vertices in $\Gamma_1$ and $\Gamma_3$.  Then by Lemma~\ref{lem:uni-morph} the morph $\langle \Gamma_1, \Gamma_3 \rangle$ is horizontal and planar.
\end{proof}

\begin{lemma}
Let $\Gamma_1,\Gamma_2$ be two planar straight-line drawings such that $\langle\Gamma_1,\Gamma_2\rangle$ is a horizontal morph. 
Then the convexity status of each angle changes at most once in the morph $\langle\Gamma_1,\Gamma_2\rangle$, i.e.,~an angle cannot change more than once between reflex and convex or vice versa.   
If additionally, every convex internal angle of $\Gamma_1$ is also 
convex in $\Gamma_2$ then the morph $\langle\Gamma_1,\Gamma_2\rangle$ is convexity-increasing.
\label{lem:convexity-inc}
\end{lemma}

\begin{proof} 
\changed{This result is a generalization of~\cite[Lemma 7]{angelini-convex-drawings-2015}, and both are proved using basic properties of unidirectional morphs from~\cite{alamdari2016morph}.}
Consider an angle formed by points $a, b, c$.  If $a$ and $c$ both lie above $b$, or both lie below $b$, then the angle maintains its convexity status (convex or reflex) during any horizontal morph.
So suppose that the ordering of the points by $y$-coordinate is $a, b, c$.  Suppose that the clockwise angle $abc$ is convex at two time points $t$ and $t'$ with $t < t'$ during the horizontal morph. 
If we add the edge $ac$ we obtain a triangle and the horizontal line through $b$ crosses $ac$ and $b$ in the same order at both time points.  Thus, the morph between $t$ and $t'$ is planar by Lemma~\ref{lem:uni-morph}, so the angle is convex at all times between $t$ and $t'$. 

We prove the second statement by contraposition.  Suppose the horizontal morph  $\langle \Gamma_1, \Gamma_2 \rangle$ is not convexity-increasing.  Then some internal angle changes from convex to reflex during the morph.  By the above statement, the angle must be convex in $\Gamma_1$ and reflex in~$\Gamma_2$.  
\end{proof}

Alamdari et al.~gave the following further condition that implies the hypothesis of Lemma~\ref{lem:uni-morph}.
We emphasize that the statement applies to planar graph drawings in general, that is, edges are not required to be straight-line.

\begin{obs}~\cite[in the proof of Lemma 13]{alamdari2016morph}
Let $\Gamma$ be a planar graph drawing of a graph~$G$ in which all faces are $y$-monotone
and let $\Gamma'$ be another planar drawing of $G$ that has the same combinatorial embedding, the same $y$-coordinates of vertices, and $y$-monotone edges. Then every line parallel to the $x$-axis crosses the same ordered sequence of edges and vertices in both drawings.
\label{lem:uni-morph2} 
\end{obs}

Note that the property of $y$-monotone faces is a necessary condition---imagine a rectangle with a stalactite from the  top and a stalagmite from the bottom such that a horizontal line through the middle of the rectangle intersects first the stalactite and then the stalagmite.  This order of intersection can be switched without changing the combinatorial embedding nor the $y$-coordinates of the vertices. 

The final ingredient we need in order to make use of the above lemmas is a way to redraw a graph to preserve the combinatorial embedding and the $y$-coordinates of the vertices, while improving convexity.  
We will again follow Alamdari et al.~\cite{alamdari2016morph} and make use of a result of Hong and Nagamochi which is described in the next section.

\subsection{Redrawing with Convex Faces while Preserving $y$-Coordinates}
\label{sec:H&N}

We build upon an 
$O(n^2)$ time
algorithm due to Hong and Nagamochi~\cite{hn-2012} that redraws a given drawing with $y$-monotone faces  such that all faces become convex while preserving the $y$-coordinates of the vertices.   
Angelini et al.~\cite{angelini-convex-drawings-2015} extended
the result to strictly convex faces by perturbing vertices to avoid angles of $\pi$.  They did not analyze the run-time.
Both~\cite{hn-2012} and~\cite{angelini-convex-drawings-2015} expressed their results in terms of \emph{level planar drawings} of \emph{hierarchical-st plane graphs}. 
Their original statements and further explanations can be found in Appendix~\ref{appendix:H&N}.

In Section~\ref{sec:Tutte}, we give a new proof of Hong and Nagamochi's result using Tutte's graph drawing 
method, which finds the vertex coordinates by solving a linear system.  
Tutte's original linear system gives rise to a symmetric matrix that can be solved quickly using the nested dissection method of Lipton et al.~\cite{Lipton}.  For our case, the matrix is not symmetric, and we need a recent generalization  of nested dissection due to Alon and Yuster~\cite{alon2013matrix}. 
Our approach results in  
an improved running time of $O(n^{1.5})$ without, and $O(n^{1.1865})$ with, fast matrix multiplication:

\begin{lemma}
\label{lem:H&N}
Let $\Gamma$ be a planar drawing of an internally 3-connected graph $G$ such that every face is $y$-monotone.
Let $C$ be a strictly convex straight-line drawing of the outer face of $G$ such that every vertex of $C$ has the same $y$-coordinate as in $\Gamma$.   Then there is a  strictly convex  straight-line drawing $\Gamma'$ of $G$ that has $C$ as the outer face and such that every vertex of $\Gamma'$ has the same $y$-coordinate as in $\Gamma$.   Furthermore, $\Gamma'$ can be found in time 
$O(n^{\omega/2})$, where $\omega$ is the matrix multiplication exponent. 
\end{lemma}

%% file: Algorithm.tex
\def\w{\ensuremath{\text{w}}_c\xspace}
\section{Computing Convexity-Increasing Morphs}
\label{sec:main-proof}

In this section we prove  Theorem~\ref{thm:morph-to-convex}.
In fact, we show multiple variants of Theorem~\ref{thm:morph-to-convex}, starting with a highly specialized version and proceeding to more and more general ones, which use the more specialized cases as building blocks.

\subsection{A Simple Case: Morphing $y$-Monotone Drawings}  
\label{sec:t1-monotone}

To give some intuition about our general proof strategy, we first consider an easy case where the 
outer face $C$
of $\Gamma$ is strictly convex and all faces are $y$-monotone.
Then we can immediately apply Lemma~\ref{lem:H&N} with the outer face fixed to obtain a new straight-line strictly convex drawing~$\Gamma'$ with all vertices at the same $y$-coordinates.
By Observation~\ref{lem:uni-morph2}, every line parallel to the $x$-axis crosses the same ordered sequence of edges and vertices in $\Gamma$ and in $\Gamma'$.  Then by Lemma~\ref{lem:uni-morph}  the morph from $\Gamma$ to $\Gamma'$ is planar.  Also it is a horizontal morph.  Thus we have a morph from $\Gamma$ to a strictly convex drawing $\Gamma'$ by way of a single horizontal morph.  Furthermore, the morph is convexity-increasing by Lemma~\ref{lem:convexity-inc} since every internal angle is convex in $\Gamma$ and $\Gamma'$.

\subsection{Morphing Drawings with a Convex Outer Face}  
\label{sec:convex-outer}

\newcommand{\outerface}[0]{{\ensuremath f_o}}
\newcommand{\outerfaceprime}[0]{{\ensuremath f^\prime_o}}

We next consider the case of a planar straight-line drawing $\Gamma$ in which the outer face is convex (but not necessarily strictly convex) and the inner faces are 
not necessarily $y$-monotone.
Additionally, assume that $\Gamma$ has no horizontal edge; we will later show how to ensure this.

\smallskip\noindent{\bf{Overview.}}
As in Section~\ref{sec:t1-monotone}, we begin by performing a horizontal morph.
Observe that such a morph preserves the local extrema and does not change their
convexity status.  
Thus the only reflex angles that can be made convex via a horizontal morph are the \emph{h-reflex} angles,
where 
an angle of inner face $f$ is called  \emph{h-reflex} if it is reflex and occurs at a vertex 
that has one neighbor in $f$ above and the other below---equivalently, the angle is reflex and is not a local extremum of~$f$.
We will show that a single horizontal morph suffices to convexify all the internal angles that are not local extrema.
The plan 
is to then conceptually ``turn the paper'' by $90^\circ$  and perform a vertical morph to make any \changed{\emph{v-reflex}} angle convex, where an angle of inner face $f$ is called a \emph{v-reflex} angle if it is reflex and occurs at a vertex that has one neighbor in $f$ to the left and the other to the right.
By continuing to alternate between the horizontal and the vertical direction, we eventually end up with a convex drawing and, thus, prove  Theorem~\ref{thm:morph-to-convex} for the case of a convex outer face.

\smallskip\noindent{\bf{A horizontal morphing step.}}
To find the desired horizontal morph we will apply Lemma~\ref{lem:H&N}. Therefore, we 
must first augment~$\Gamma$ to a drawing with $y$-monotone faces by inserting $y$-monotone edges which are not necessarily straight-line.
For an example see Fig.~\ref{fig:exLocExReflex}.  
This is a standard operation in upward planar (or ``monotone'') drawing~\cite[Lemma 4.1]{DiBattista-st-planar}~\cite[Lemma 3.1]{pach2004monotone}, but we need the stronger property that the new edges are only incident 
to local extrema; otherwise we would relinquish control of convexity at that vertex.   We will use:

\begin{lemma}\label{lem:decomp}
Let $\Gamma$ be straight-line planar drawing  of an internally 3-connected graph. Then $\Gamma$ can be augmented by adding edges to obtain a drawing $\Gamma'$ such that each additional edge is a $y$-monotone curve joining two local extrema of some face in $\Gamma$, every inner face is $y$-monotone and the augmented graph is internally $3$-connected.
Furthermore, the additional edges can be found in $O(n \log n)$ time.
\end{lemma}

\begin{proof} 
Recall that a face is $y$-monotone if and only if it has no reflex local extrema.
Our proof is by induction on the number of reflex local extrema in all inner faces of the drawing.
If there are none, then all inner faces are $y$-monotone.
Otherwise, consider an inner 
face $f$ that has a local extremum $u$.  For an illustration consider  Fig.~\ref{fig:exLocExReflex}.
%
 \begin{figure}[bt]
  \centering
  \includegraphics{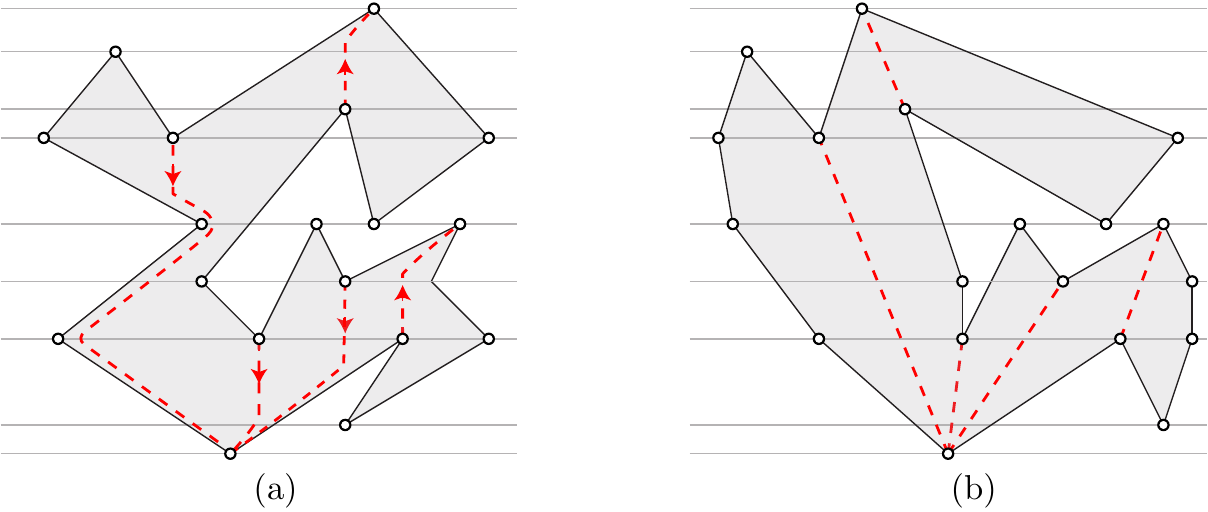}
  \caption{(a) A face that is not $y$-monotone.   The dashed edges inside the face are added by Lemma~\ref{lem:decomp}.
  (b) The face after application of Lemma~\ref{lem:morph-A}.}
  \label{fig:exLocExReflex}
 \end{figure}
%
Assume without loss of generality that $u$ is a local minimum. 
We want to find a local extremum $v$ below $u$ such that we can insert a $y$-monotone curve from $v$ to $u$ within face $f$.
To construct the curve go vertically downwards from $u$ to the first point $p_u$ on the boundary of $f$, 
and then follow a $y$-monotone chain of $f$'s boundary downwards from $p_u$ to a local minimum $v$.
Adding the edge $(u,v)$ divides $f$ into two faces, and decreases the total number of  local extrema. 
\changed{We note that Pach and T\'oth~\cite{pach2004monotone} used a similar idea to triangulate with monotone curves, although their curves stopped at the first vertex on $f$'s boundary and we must continue to the first local minimum.} Since we will only insert inner edges, the graph remains internally 3-connected.

To complete the proof we briefly describe how the set of augmenting edges can be found in time $O(n \log n)$. 
We deal with the local reflex minima; the maxima can be dealt with in a second phase.
Find a trapezoidization of the drawing~\cite{CompGeom} in $O(n \log n)$ time.  This gives the point $p_u$ for each local reflex minimum $u$ in face $f$.  
We can preprocess the graph in time $O(n)$ to find, for each edge $e$ in face $f$, the local minimum $v$ that is reached by following a $y$-monotone chain downward from $e$ in $f$. 
This gives the set of augmenting edges.  We must still find the cyclic order of augmenting edges incident with each vertex $v$.  We separate into those that arrive at $v$ from the left and those that arrive from the right.  Within each of these sets, we sort by the $y$-coordinate of points $p_u$. 
 \end{proof}

This observation allows us to prove the following:

\begin{lemma}
\label{lem:morph-A}
 Let $\Gamma$ be a straight-line planar drawing of an internally 3-connected graph with a convex outer face and no horizontal edge.  There exists a horizontal planar morph to a straight-line drawing $\Gamma'$ with a strictly convex outer face and every internal angle that is not a local extremum is strictly convex in $\Gamma'$.  Furthermore, the morph $\langle\Gamma,\Gamma'\rangle$ is convexity-increasing, and can be found in time $O(n^{\omega/2})$, where $\omega$ is the matrix multiplication exponent.
\end{lemma}
\begin{proof} Use Lemma~\ref{lem:decomp} to augment $\Gamma$ with a set of edges $A$ such that $\Gamma \cup A$ is a planar drawing in which all faces are $y$-monotone, and any edge of $A$ goes between two local extrema in some inner face.  This takes $O(n \log n)$ time.
Let $C$ be the outer face of $\Gamma$.  Create a new drawing  $C'$ of $C$ such that $C'$  is strictly convex and preserves the $y$-coordinates of vertices.  

By Lemma~\ref{lem:H&N} with the outer face $C'$ we obtain (in time $O(n^{\omega/2})$) a new straight-line strictly convex drawing $\Gamma' \cup A'$ with all vertices at the same $y$-coordinates as in $\Gamma$.  (Here $A'$ is a set of straight-line edges corresponding to $A$.) 
By Observation~\ref{lem:uni-morph2} every line parallel to the $x$-axis crosses the same ordered sequence of edges and vertices in $\Gamma \cup A$ and in $\Gamma' \cup A'$. 
Then by Lemma~\ref{lem:uni-morph} the morph from $\Gamma$ to $\Gamma'$ is a planar horizontal morph. 

Any internal angle of $\Gamma$ that is not a local extremum has no edge of $A$ incident to it, and thus becomes strictly convex in $\Gamma'$.   Any internal angle of $\Gamma$ that is a local extremum maintains its convex/reflex status in $\Gamma'$.  Thus by Lemma~\ref{lem:convexity-inc} the morph is convexity-increasing.
The run-time to find the morph (i.e.,~to find $\Gamma'$) is $O(n^{\omega/2})$.
\end{proof}

\smallskip\noindent{\bf{Making progress.}}
Lemma~\ref{lem:morph-A} generalizes to directions other than the horizontal direction:   for any direction~$d$ that is not parallel to an edge of~$\Gamma$, there exists a convexity-increasing unidirectional (with respect to~$d$) morph that convexifies all internal angles that are not extreme in the direction orthogonal to~$d$.
Thus, if we do not insist on a sequence of \emph{horizontal and vertical} morphs, we immediately obtain a proof of Theorem~\ref{thm:morph-to-convex} for the case of a convex outer face, since, for each reflex angle, we can choose a direction $d$ that convexifies it.

In order to prove the stronger result 
that uses only two orthogonal directions, we need to alternate between the horizontal and the vertical direction.
Note that after one application of Lemma~\ref{lem:morph-A}, we do not necessarily obtain a drawing which contains a v-reflex vertex, see Fig.~\ref{fig:exHorizontalMorph}(a,b).
In order to ensure that our algorithm makes progress after every step, we  prove a strengthened version of Lemma~\ref{lem:morph-A}, which ensures that there is at least one h-reflex or v-reflex vertex
available after each step.

\begin{lemma}
\label{lem:morph-B}
 Let $\Gamma$ be a straight-line planar drawing of an internally 3-connected graph with a convex outer face and no horizontal edge.  There exists a horizontal planar morph to a straight-line drawing $\Gamma''$ such that 
 \begin{compactitem}
   \item[(i)] the outer face of $\Gamma''$ is strictly convex,
   \item[(ii)] every internal angle that is not a local extremum is convex in $\Gamma''$,
   \item[(iii)] $\Gamma''$ has no vertical edge, and
   \item[(iv)] if $\Gamma''$ is not convex, then it has at least one v-reflex angle.
\end{compactitem} 
Furthermore, the morph is convexity-increasing, and can be found in time $O(n^{\omega/2})$.
\end{lemma} 

\begin{proof}
We first apply Lemma~\ref{lem:morph-A} to obtain a morph from $\Gamma$ to a drawing $\Gamma'$ that satisfies (i) and (ii). 
If $\Gamma'$ satisfies all the requirements, we are done.  
Otherwise we will achieve the properties (iii) and (iv) by shearing the drawing $\Gamma'$. 
Eliminating vertical edges via a horizontal shear is easy, so we concentrate on the requirement (iv) about v-reflex angles.
Suppose $\Gamma'$ is not convex and has no v-reflex angle.
Consider any reflex angle of $\Gamma'$, say at vertex $u$ in the inner face~$f$.
By property (ii), $u$ must be a local extremum in $\Gamma'$; otherwise it would be convex.  
We will apply a horizontal shear transformation to create a drawing $\Gamma''$ in which the angle at $u$ becomes v-reflex, i.e.,~in which the $x$-coordinate of $u$ is between the $x$-coordinates of its two neighbors in~$f$.
Furthermore, the shear should eliminate all vertical edges, e.g.,~see Fig.~\ref{fig:exHorizontalMorph}.
The shear can be found in linear time.

 \begin{figure}[htb]
  \centering
  \includegraphics{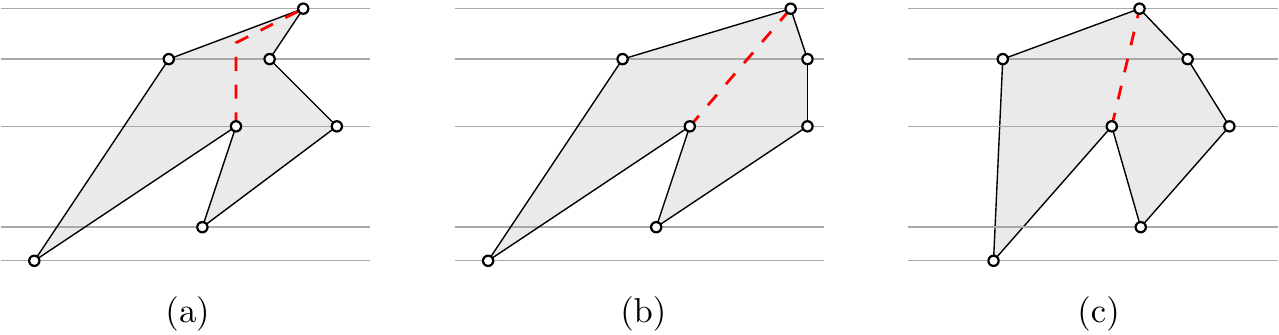}
  \caption{(a) A face that is not $y$-monotone. 
  (b) The face after application of Lemma~\ref{lem:morph-A}.  There is a vertical edge and the single reflex vertex is not v-reflex. (c) After applying a horizontal shear transformation, the reflex vertex is v-reflex and there are no vertical edges.
 }
  \label{fig:exHorizontalMorph}
 \end{figure}

Since shearing is an affine transformation, $\Gamma''$ has the same convex/reflex angles as $\Gamma'$.  Thus $\Gamma''$ satisfies all the properties.
The linear morph $\langle \Gamma' , \Gamma'' \rangle$ is a planar horizontal  morph that preserves the convex/reflex status of every angle.  
By Lemma~\ref{lem:union} the morph $\langle \Gamma, \Gamma'' \rangle$ is a horizontal planar morph.  By Lemma~\ref{lem:convexity-inc} it is convexity-increasing. 
The morph can be found in time $O(n^{\omega/2})$.
\end{proof}

We are now ready to prove Theorem~\ref{thm:morph-to-convex} for the case of a convex outer face.

\begin{proof}[\proofof{Theorem~\ref{thm:morph-to-convex} for the case of a convex outer face}]
If the given drawing $\Gamma$ has a horizontal edge and/or if~$\Gamma$ has internal reflex angles such that none of them is h-reflex, then we use one vertical shear as in the proof of Lemma~\ref{lem:morph-B} to remedy this.
Then, in the special case that there is no internal reflex angle, but there are angles of degree exactly~$180^\circ$, we may apply Lemma~\ref{lem:morph-B} once to obtain the desired strictly convex drawing.
Otherwise, there are internal reflex angles; and we apply Lemma~\ref{lem:morph-B} alternately in the horizontal and vertical directions until the drawing is strictly convex.
In each step there is at least one h-reflex or v-reflex angle that becomes convex.
Thus, the number of morphing steps is at most $\mathrm{max}\lbrace 2,r+1\rbrace \le n$, where $r$ is the number of inner reflex angles in $\Gamma$. The resulting total run-time is $O(n^{1 + \omega/2})$. 
 \end{proof}

%% file: OuterFace.tex
\subsection{Morphing Drawings of $3$-Connected Graphs}
 \label{sec:non-convex-outer}

In this section we prove the case of Theorem~\ref{thm:morph-to-convex} where the outer face of~$\Gamma$ is not convex. However, we will assume that the given graph~$G$ is $3$-connected (instead of only \emph{internally} $3$-connected).

\noindent{\bf Overview.} 
On a high level, our approach works as follows:
We 
first augment 
the outer face of~$\Gamma$ with edges from its convex hull to obtain a drawing of 
an augmented graph with a convex outer face.
We then apply the results from Section~\ref{sec:convex-outer} to morph to a strictly convex drawing
and then remove the extra edges on the outer face one-by-one.
After each removal of an edge, we morph to a strictly convex drawing of the reduced graph using at most three horizontal or vertical morphs.

\noindent{\bf Augmenting the outer face.} 
Compute the convex hull of $\Gamma$.  Any segment of the convex hull that does not  correspond to an edge of $G$ becomes a new edge that we add to $G$.  Let $A$ denote the new edges and $G \cup A$ denote the augmented graph with straight-line planar drawing $\Gamma \cup \Gamma_A$.
Note that adding edges maintains 3-connectivity.
Each edge~$e\in A$ is part of the boundary of an inner face~$f_e$ of~$\Gamma \cup \Gamma_A$. We call~$f_e$ the \emph{pocket} of $e$.
We apply the result of Section~\ref{sec:convex-outer} to obtain a strictly convex drawing  of~$G \cup A$, see Fig.~\ref{fig:outerFace3con}(a). 
Note that the techniques used in that section give us a drawing with no horizontal or vertical edges.

We remark that this step is the reason why we
limit ourselves to $3$-connected graphs in this subsection:
Adding the convex hull edges in a drawing of an \emph{internally} $3$-connected graph may create non-external separations pairs, see Fig.~\ref{fig:int3conPe}(a).
This would prevent us from using the algorithm from Section~\ref{sec:convex-outer} as this algorithm uses Lemma~\ref{lem:H&N} which requires the input graph to be internally $3$-connected.

\noindent{\bf Popping a pocket outward.} 
In this final step, we describe a way to remove an edge of $A$ and ``pop'' out the vertices of its pocket so that they become part of the convex hull.
Lemma~\ref{lem:H&N}
serves once again as an important subroutine.
We make ample use of the fact that we may freely specify the desired subdrawing of the outer face after each application of Lemma~\ref{lem:H&N},
as long as we maintain either the $x$-coordinates or the $y$-coordinates of all vertices.

We remark that the following lemma applies to \emph{internally} $3$-connected graphs, not only $3$-connected graphs.
We plan to use it in the following section, in which we prove Theorem~\ref{thm:morph-to-convex} in its general form.
In fact, the final algorithm will use the entire procedure described in this section as a subroutine (we will augment the internally $3$-connected graph such that adding the convex hull edges does not create non-external separation pairs).

\begin{lemma}
Let $\Gamma$ be a strictly convex drawing of an internally $3$-connected graph $G$ without vertical edges and let $e$ be an edge on the outer face.  If $G - e$ is internally 3-connected, then $\Gamma - e$ can be morphed to a strictly convex drawing of $G-e$ without vertical edges via at most \changed{three convexity-increasing morphs, each of which is horizontal or vertical.} Furthermore, the morphs can be found in  time $O(n^{\omega/2})$.
\label{lem:outerFace3con}
\end{lemma}

\begin{proof}  Our morph will be specified by a sequence of drawings, 
$\Gamma$, $\Gamma_1$, $\Gamma_2$, $\Gamma_3$, where the first and the last morph are vertical and the second morph, which we can sometimes skip, is horizontal.

Let $e=(u,v)$.  We first perform a vertical morph from $\Gamma$ to a strictly convex drawing $\Gamma_1$  in which  vertex $u$ is top-most or bottom-most and which does not contain vertical or horizontal edges.
Since $\Gamma$ is strictly convex and has no vertical edges, it is also $x$-monotone.
Therefore, the desired drawing~$\Gamma_1$ can be found by choosing some strictly convex drawing of the outer face in which $u$ is extreme while
maintaining the $x$-coordinates of all vertices,
and then using one application of Lemma~\ref{lem:H&N} (for vertical morphs).
Additional, we may need to apply a vertical shearing transformation in order to get rid of horizontal edges.
This is easily done while still guaranteeing that~$u$ is extreme in the~$y$-direction.
Analogous to Section~\ref{sec:t1-monotone}, by combining Observation~\ref{lem:uni-morph2}, Lemma~\ref{lem:uni-morph} and Lemma~\ref{lem:convexity-inc} we conclude that the horizontal morph $\langle \Gamma,\Gamma_1\rangle$ is planar and convexity-increasing.

For the remainder of the proof, assume without loss of generality that $u$ is the top-most vertex and that $v$ lies to the right of $u$ in $\Gamma_1$.  The other cases are symmetric.
Let~$p_{uv}$ 
denote the path from  $u$ to $v$ in $f_e -e$. We distinguish two cases 
depending on the shape of~$f_e$ in $\Gamma_1$.

\noindent\textbf{Case 1:} 
The path~$p_{u,v}$ is
$x$-monotone in~$\Gamma_1$, see Fig.~\ref{fig:outerFace3con}(b).
In this case we can skip the second step of the morph sequence. 
We will remove~$e$ and compute a vertical morph from $\Gamma_1-e$ to a strictly convex drawing $\Gamma_3$ of~$G-e$ without vertical edges.  Once again, this can be done by combining Lemma~\ref{lem:H&N} (for vertical morphs),  Observation~\ref{lem:uni-morph2}, Lemma~\ref{lem:uni-morph} and Lemma~\ref{lem:convexity-inc} as long as we can specify 
a strictly convex drawing of the outer face of $\Gamma_1 - e$ in which the~$x$-coordinates match those of~$\Gamma_1$.
It suffices to compute a suitable new reflex chain for~$p_{uv}$, see Fig.~\ref{fig:outerFace3con}(b).
\begin{figure}[tb]
  \centering
\ifeurocg
  \includegraphics{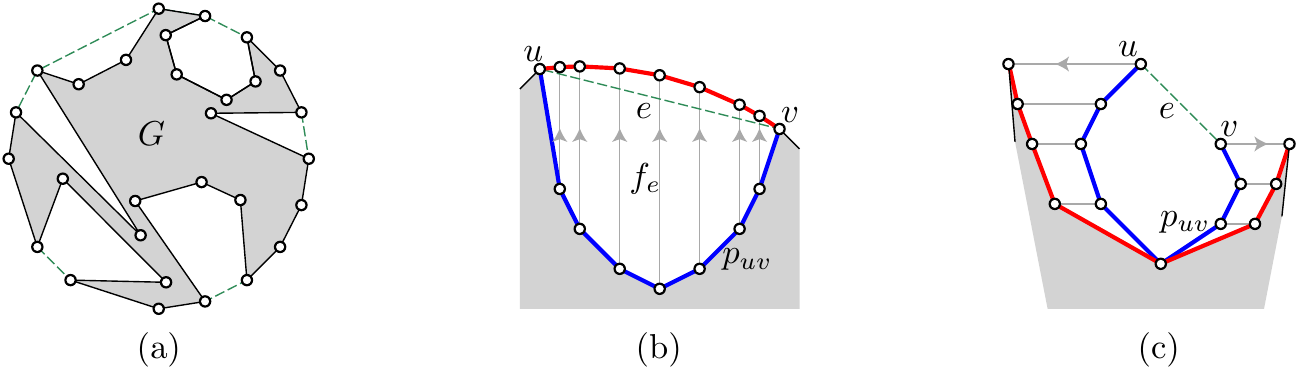}
\else
  \includegraphics{fig/outerFace3con}
\fi
  \caption{(a) Schematic of the convex drawing of~$G \cup A$. Graph $G$ is depicted in gray, edges of $A$ are dashed, and the pockets are white.  
  (b)--(c) Cases 1 and 2 for Lemma~\ref{lem:outerFace3con}, where faint gray arrows indicate explicit placements on the convex hull.
}
  \label{fig:outerFace3con}
 \end{figure}

\noindent\textbf{Case 2:} 
The path~$p_{uv}$ is not $x$-monotone.
In this case we will compute a horizontal morph from $\Gamma_1$ to  
a strictly convex drawing $\Gamma_2$ in which $u$ and~$v$ are the unique left-most and the unique right-most vertices and which does not contain a vertical edge.  
Once again, this can be done by combining Lemma~\ref{lem:H&N}, a horizontal shearing transformation, Observation~\ref{lem:uni-morph2}, Lemma~\ref{lem:uni-morph} and Lemma~\ref{lem:convexity-inc} as long as we can specify 
a strictly convex drawing of the outer face of $\Gamma_1$ in which~$u$ and~$v$ are the left-most and right-most vertices and the $y$-coordinates match those of~$\Gamma_1$.  This is possible because~$u$ is top-most, see Fig.~\ref{fig:outerFace3con}(c).

In the drawing $\Gamma_2$ 
the pocket $f_e$ is convex with extreme points $u$ and $v$ so the path~$p_{uv}$ is $x$-monotone and, hence, by Case~1, there is a vertical morph from $\Gamma_2 -e$ to a strictly convex drawing $\Gamma_3$ of $G-e$.
\end{proof}

Observe that each application of Lemma~\ref{lem:outerFace3con} increases the number of 
vertices
of $G$ on the convex hull.
Thus, the proof of Theorem~\ref{thm:morph-to-convex} follows by induction on the number of convex hull vertices.
Let~$\rho$ denote the number of pockets of the strictly convex drawing of~$G\cup A$ obtained by applying the algorithm from Section~\ref{sec:convex-outer}.
We observe that~$\rho \le n/2$, since each pocket can be associated with two private vertices of $G$: the clockwise first of the two convex hull vertices defining the pocket and its clockwise successor, which is private to the pocket.

We use at most three horizontal and vertical morphs to pop out a pocket.
Hence, the number of morphs needed to deal with all the pockets can be bounded by~$2\rho +1$ by observing that each application of Lemma~\ref{lem:outerFace3con} involves a vertical-horizontal-vertical morph sequence, and \ifeurocg we can \else Lemma~\ref{lem:union} allows us to \fi{}compress two consecutive vertical morphs into one.
Additionally, obtaining the strictly convex drawing of~$G\cup A$ requires at most~$\mathrm{max}\lbrace 2,r+1\rbrace$ morphing steps, where~$r$ is the number of internal reflex angles in the drawing~$\Gamma\cup \Gamma_A$.
This number can be bounded by $\mathrm{max}\lbrace 2,r+1\rbrace\le n-\rho +1$ since each pocket can be associated with a private vertex of the convex hull, that cannot have an internal reflex angle.
Hence, the total number of morphs is at most $n-\rho +1 + 2\rho +1=n+\rho+2\le 1.5n+2$, where the last inequality uses the fact that~$\rho \le n/2$.
The run time of the algorithm is $O(n^{1 + \omega/2})$.

%% file: int3con.tex
\subsection{Morphing  Drawings of Internally 3-Connected Graphs}
 \label{sec:int3con}
  
 So far,  if the given drawing does not have a convex outer face, we have restricted our attention to the class of 3-connected graphs.
 The only reason why the approach described in Section~\ref{sec:non-convex-outer} is not able to handle the case of \emph{internally} $3$-connected graphs is that the addition of the convex hull edges~$A$ may create non-external separation pairs, see Fig.~\ref{fig:int3conPe}(a). As a result, the augmented drawing~$\Gamma_G\cup \Gamma_A$ is no longer a valid input for Lemma~\ref{lem:H&N}.
 In this section we extend our algorithm such that it is able to convexify drawings of internally 3-connected graphs and, thus, we conclude the proof of  Theorem~\ref{thm:morph-to-convex} in its general form.
 
   \begin{figure}[bht]
  \centering
  \includegraphics[page=1]{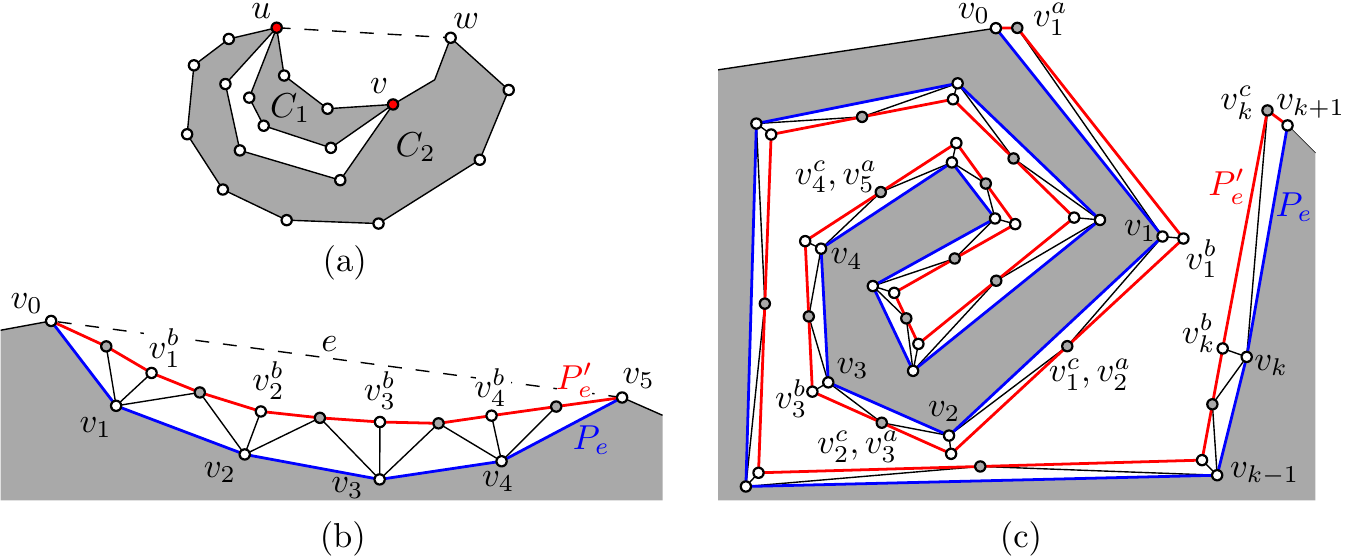}
  \caption{(a) An internally 3-connected graph with an external separation pair~$\lbrace u,v\rbrace $. Adding the convex hull edge $uw$ turns $\lbrace u,v\rbrace $ into a non-external separation pair both since $v$ becomes an internal vertex and since $C_1$ no longer has a vertex on the outer face.
  (b) Schematic drawing of the path $P_e'$ created for the pocket defined by the convex hull edge $e=\lbrace v_0,v_5\rbrace$.
  (c) Geometrically, we embed $P_e'$ very close to $P_e$. This is possible regardless of the shape of the pocket.
}
  \label{fig:int3conPe}
 \end{figure}
 
 \noindent{\bf Overview.} Let~$\Gamma$ be a straight-line planar drawing of an internally 3-connected graph~$G=(V,E)$.
 As a first step, we augment the outer face of~$\Gamma$ by adding new edges and vertices near each pocket.
 The goal of this step is to ensure that we can add convex hull edges without introducing non-external separation pairs.
 We then apply the algorithm from Section~\ref{sec:non-convex-outer}, which results in a drawing in which all the new vertices appear on the strictly convex outer face.
 Finally, we remove the new vertices one-by-one; gradually turning the strictly convex drawing of the augmented graph into a strictly convex drawing of~$G$.

\input{int3con-insert}

%% file: int3con-insert.tex
In slightly more detail, our method involves the following steps:
 
\noindent{\bf Step 1: Augmenting the outer face.}  
 We will augment the graph $G$ to $G'$, and then augment the drawing $\Gamma$ to $\Gamma'$, 
 which involves geometric arguments.
 The goal is to ensure that $G'$ is internally 3-connected and has the additional property that adding convex hull edges to $\Gamma'$ does not introduce non-external separation pairs, as, for example, in Fig.~6(a).  We will show that bad cases only arise when a pocket has a vertex that is part of a separating pair.  Thus, our method will be to add an extra ``bufffer'' layer of vertices to each pocket boundary, while ensuring that each buffer vertex is not part of a separating pair.  
   
\noindent{\bf Step 2: Convexifying the augmented drawing.}
 The goal of this step is to find a convexity-increasing morph from the drawing $\Gamma'$ of the augmented graph~$G'$ to a strictly convex drawing of~$G'$.
In principal, the idea for accomplishing this task is very simple: we just apply the algorithm described in Section~\ref{sec:non-convex-outer} to $\Gamma'$.
The challenging aspect is, that this algorithm uses Lemma~\ref{lem:H&N} as a subroutine and, thus, we need to ensure internal 3-connectivity of the input graph.

\noindent{\bf Step 3: Removing the additional vertices.}
 At this point we have a strictly convex drawing of $G'$.  We must now 
``reverse'' the augmentation process, removing vertices of $G'$ to get back to $G$. 
 After each vertex is removed we will morph to obtain a strictly convex drawing.
 As in Step 2, the graph must remain internally 3-connected at each step.
 We will therefore treat the augmentation process of Step 1 as an iterative process, adding vertices one-by-one and ensuring
 that the graph is internally 3-connected at each step.

\medskip
We now give further details on these three steps.

\noindent{\bf Step 1: Augmenting the outer face.} 
 We compute the convex hull of $\Gamma$.
 Let~$e$ be a convex hull edge with~$e\notin E$.
 The following steps are illustrated in Fig.~\ref{fig:int3conPe}(b).
 Let~$P_e=(v_0, v_1,\dots,v_{k+1})$ be the unique path on the outer face of~$\Gamma$ such that~$P_e+e$ is a cycle with~$G-P_e$ in its exterior.
 We introduce~$2k+1$ new vertices that form a path~$P_e'$ connecting~$v_0$ and~$v_{k+1}$:
 \begin{align*}
P_e'=(v_0,~v_1^a,v_1^b,v_1^c~=~v_2^a,v_2^b,v_2^c~=~v_3^a,\dots ,v_{k-1}^c~=~v_k^a,v_k^b,v_k^c,~v_{k+1})
\end{align*}
Note that every inner vertex~$v_i$ of~$P_e$ gets a ``private copy'' $v_i^b$ in~$P_e'$.
Two consecutive copies~$v_i^b$ and~$v_{i+1}^b$ are connected via another vertex which is equipped with two labels $v_i^c=v_{i+1}^a$, which will simplify the notation later on.
Additionally, we add the edges $\lbrace v_i,v_i^a\rbrace$, $\lbrace v_i,v_i^b\rbrace$ and $\lbrace v_i,v_i^c\rbrace$ for $i=1,\dots,k$.

Geometrically, the new path~$P_e'$ is embedded in a planar fashion very close to~$P_e$, see Fig.~\ref{fig:int3conPe}(c). This can be accomplished regardless of the shape of~$P_e$:
Assume that~$k\ge 2$ and let~$\varepsilon$ be the smallest distance between any
pair of disjoint edges on the cycle~$P_e+e$.
For $i=1,\dots,k$, we place~$v_i^b$ 
on the angular bisector of the outer angle at $v_i$ such that its distance to~$v_i$ is smaller than~$\varepsilon/2$.
If $i>1$, the vertex~$v_i^a$ is placed in the center of the line-segment $v_{i-1}^bv_i^b$.
Similarly, if~$i<k$, the vertex~$v_i^c$ is placed  in the center of the line-segment $v_{i}^bv_{i+1}^b$.

We make sure that during this procedure no vertex is placed in the exterior of the convex hull of~$\Gamma$; with the following exceptions:
the vertices~$v_1^a$ and~$v_k^c$ play a special role and are placed close to~$v_0$ and~$v_{k+1}$, respectively, such that they \emph{do} appear on the convex hull of the augmented drawing.

Note that for the special case of small pockets with~$k=1$, the value~$\varepsilon$ would not be well-defined, as there are no disjoint edges on~$P_e+e$.
However, in this case it is easy to directly compute an embedding of~$P_e'$ with the desired properties (i.e.~planarity; and only $v_1^a$ and $v_k^c$ appear on the convex hull).

We repeat the process for all convex hull edges~$e\notin E$ of~$\Gamma$ and use $G'=(V',E')$ and~$\Gamma'$ to denote the resulting plane graph and drawing respectively.
The total run-time for this step is dominated by the time to compute the values~$\varepsilon$, which can be done in~$O(n^2)$ total time.
We remark that this step could be implemented more efficiently, for example by computing a medial axis~\cite{DBLP:journals/dcg/ChinSW99}. However, since our total runtime is $\omega(n^2)$, we refrain from stating the details.

We will now prove that $G'$ is internally 3-connected.  Keeping in mind our plan for Step 3, we will add the vertices of $G'$ one-by-one, showing that each addition preserves the property of being internally 3-connected.
We begin with two basic operations that preserve internal 3-connectivity.  

\begin{lemma} Let $H$ be an internally 3-connected graph with an edge $(a,b)$ on the outer face.  Construct $H'$ by adding a new vertex $x$ in the outer face connected to $a$ and $b$.  Then $H'$ is internally 3-connected.
\label{claim:add-V}
\end{lemma}
\begin{proof}
By internal $3$-connectivity, Property~\ref{I3} of Lemma~\ref{lem:int3conDefs} holds for~$H$.
Thus, it is clear that~$H'$ also satisfies Property~\ref{I3} and, hence,~$H'$ is internally $3$-connected.
\end{proof}
\begin{lemma} Let $H$ be an internally 3-connected graph with two consecutive edges $(a,b)$ and $(b,c)$ on the outer face.  Construct $H'$ by adding a new vertex $x$ in the outer face connected to $a, b$ and $c$.  Then $H'$ is internally 3-connected.
\label{claim:add-claw}
\end{lemma}
\begin{proof}
By internal $3$-connectivity, Property~\ref{I3} of Lemma~\ref{lem:int3conDefs} holds for~$H$.
It suffices to show that Property~\ref{I3} also holds for~$H'$.
Let~$f_0$ and~$f_0'$ denote the outer faces of~$H$ and~$H'$, respectively.
Clearly,~$b$ has three paths to~$f_0'$ that are disjoint except for~$b$.
So let~$v\neq b$ be some internal vertex of~$H'$ and note that~$v$ is also internal in~$H$.
Hence, by Property~\ref{I3} of~$H$,~$v$ has three paths to~$f_0$ that are disjoint except for~$v$.
At most one of these paths does not end at~$f_0'$, namely if its endpoint on~$f_0$ is~$b$.
However, appending the edge~$\lbrace b,x\rbrace$ to this paths yields the desired three paths from~$v$ to~$f_0'$ that are disjoint except for~$v$.
Hence, Property~\ref{I3} holds for~$H'$.
\end{proof}

With these operations in hand, we can show that $G'$ is internally 3-connected, and---more strongly---that we can build $G'$ by adding one vertex at a time, preserving internal 3-connectivity.
Let~$V'=V\cup V^b\cup V^{ac}$ where $V^b$ is the set of all
vertices whose upper index is $b$, and where $V^{ac}$ is the set of the
remaining vertices (whose upper index is $a$ and/or $c$). 

\begin{lemma}\label{lem:removeLabelB}
Starting with $G$ and adding 
the vertices of $V^{ac}$ one-by-one in any order and then the vertices of $V^b$ one-by-one in any order 
produces an internally 3-connected graph at each step.
\end{lemma}

\begin{proof}
The addition of each vertex of $V^{ac}$ maintains internal 3-connectivity by
Lemma~\ref{claim:add-V}.
The addition of each vertex of $V^b$ maintains internal 3-connectivity
by Lemma~\ref{claim:add-claw}.
\end{proof}

\noindent{\bf Step 2: Convexifying the augmented drawing.}
As mentioned above, the plan is to apply the algorithm described in Section~\ref{sec:non-convex-outer} to $\Gamma'$.
That algorithm 
 adds the convex hull edges of~$\Gamma'$ and then iteratively removes these edges while performing some morphing steps before and after each removal.
Each morphing step requires one application of Lemma~\ref{lem:H&N}.
Therefore, in order to prove the correctness of Step 2, we need to ensure that before and after each removal of a convex hull edge, the input graph is internally 3-connected.

We begin by observing that the new vertices of $G'$ are not part of separation pairs.  Then we show that adding a convex hull edge is safe when none of the vertices of its pocket are in separating pairs.

\begin{obs}
\label{obs:NEWnoNewVertexSep}
No vertex of $V'\setminus V$ is in a separating pair of $G'$.
\end{obs}
\begin{proof}
By construction, removing any one of these vertices can not introduce a cut vertex.
\end{proof}

We now give one more operation that preserves internal 3-connectivity.

\begin{lemma} Let $H$ be an internally 3-connected graph, with vertices $a$ and $b$ on the outer face.  Let $P$ be one of the paths from $a$ to $b$ along the outer face.  Assume that no vertex of $P$ is part of a separating pair in $H$.  
Let $H'$ be the result of adding the edge $(a,b)$ embedded such that $P$ becomes internal.  Then $H'$ is internally 3-connected.
\end{lemma}
\begin{proof}
Suppose that~$H'$ has a separation pair~$\lbrace u,v\rbrace$.
This pair is also separating in~$H$.
Further, since~$H$ is internally $3$-connected,~$\lbrace u,v\rbrace$ is an \emph{external} separation pair of~$H$.
Since both~$u$ and~$v$ do not belong to~$P$, Observation~\ref{lem:externalSepPair} implies that~$a$ and~$b$ belong to the same component of $(H-u-v)$.
Thus, is easy to verify that adding the edge between~$a$ and~$b$ maintains the six conditions of Observation~\ref{lem:externalSepPair} (in particular, Property~\ref{E2} holds as~$a$ and~$b$ remain on the outer face).
Therefore, $\lbrace u,v\rbrace$ is an external separation pair of~$H'$.
\end{proof} 

Recall that the construction of~$\Gamma'$ ensures that all convex hull edges~$e'$ which do not correspond to edges of~$G'$ have the form $e'=\lbrace v_1^a,v_k^c\rbrace$, where $v_1^a,v_k^c$ are the first and last internal vertex of one of the paths~$P_e'$.
By Observation~\ref{obs:NEWnoNewVertexSep} the interior vertices of~$P_e'$ form a path~$P$ on the outer face which does not contain any vertices that are part of a separation pair. Further, adding the edge~$e'$ encloses~$P$ in an internal face. Thus, we obtain:

\begin{corollary}
\label{cor:canAddCh}
Let~$A$ denote the set of convex hull edges of~$\Gamma'$ which do not correspond to edges of~$G'$.
Then, for any~$S\subseteq A$ the plane graph~$G'+S$ is internally 3-connected.
\end{corollary}

\noindent{\bf Step 3: Removing the additional vertices.}
At this point we have a strictly convex drawing of $G'$ and want to convert it to a strictly convex drawing of $G$.
We will remove the vertices of~$V^b$ iteratively; one-by-one.
After each such removal, we will perform up to two morphing steps in order to recover a strictly convex drawing of the reduced graph.
Recall that by Lemma~\ref{lem:removeLabelB}, all the intermediary graphs are internally $3$-connected and, thus, they are valid inputs for Lemma~\ref{lem:H&N}.

\begin{lemma}\label{lem:removeLabelBmorph}
Let $B\subseteq V^b$ and let~$\Gamma'_B$ be a strictly convex drawing of~$G'-B$ without vertical or horizontal edges.
Further, let~$v_i^b\in V^b\setminus B$.
Then, there is a convexity-increasing morph from~$\Gamma'_B-v_i^b$ to a strictly convex drawing~$\Gamma''$ of $G'-B-v_i^b$ without vertical or horizontal edges.
Moreover, there is such a morph which consists of a sequence of up to $2$ horizontal / vertical morphs.
The morphing sequence can be found in $O(n^{\omega /2})$ time.
\end{lemma}

\begin{proof}
Without loss of generality, we may assume that~$v_i^c$ is located to the bottom-right of~$v_i^a$ and that~$v_i^b$ is located to the right of the oriented line $\overrightarrow{v_i^cv_i^a}$, see Fig.~\ref{fig:int3conRmVertices}(a).
We distinguish four cases regarding the position of the vertex~$v_i$, for an illustration see Fig.~\ref{fig:int3conRmVertices}(b).
 \begin{figure}[htb]
  \centering
  \includegraphics[page=2]{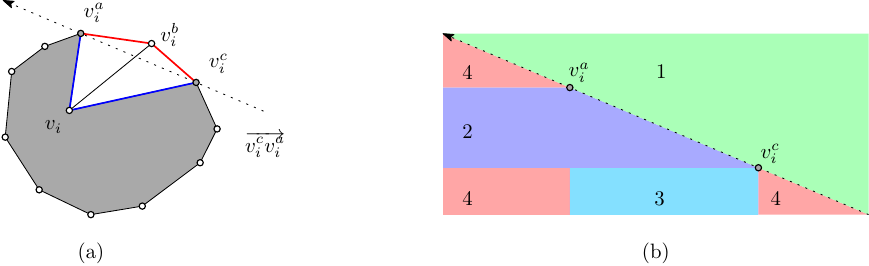}
  \caption{(a) In Step 3, we iteratively remove the vertices $v_i^b$ causing their counterparts $v_i$ to become part of the outer face. (b) The regions corresponding to the Cases 1--4.
}
  \label{fig:int3conRmVertices}
 \end{figure}

\noindent{\bf Case 1:} $v_i$ is located to the right of the oriented line $\overrightarrow{v_i^cv_i^a}$. We observe that the subdrawing~$\Gamma_B'-v_i^b$ is already a convex drawing of~$G'-B-v_i^b$, so there is nothing to show.

\noindent{\bf Case 2: } $\textrm{y}(v_i^a)>\textrm{y}(v_i)>\textrm{y}(v_i^c)$; and we are not in Case 1.
Let $C$ be a strictly convex drawing of the outer face of~$G'-B-v_i^b$, such that  every vertex in $C$ has the same~$y$-coordinate as in~$\Gamma'_B-v_i^b$ (we can easily find such a drawing~$C$ by adding~$v_i$ to the convex hull of~$\Gamma'_B-v_i^b$ in a strictly convex fashion).
Then, Lemma~\ref{lem:H&N} applied to~$\Gamma'_B-v_i^b$ and~$C$ (potentially followed by a horizontal shearing transformation in order to remove vertical edges), yields the desired drawing~$\Gamma''$.
Analogous to Section~\ref{sec:t1-monotone}, by combining Observation~\ref{lem:uni-morph2}, Lemma~\ref{lem:uni-morph} and Lemma~\ref{lem:convexity-inc} we conclude that the horizontal morph $\langle \Gamma'_B-v_i^b,\Gamma''\rangle$ is planar and convexity-increasing.

\noindent{\bf Case 3: } $\textrm{x}(v_i^a)<\textrm{x}(v_i)<\textrm{x}(v_i^c)$; and we are not in Case 1 or Case 2.
Let $C$ be a strictly convex drawing of the outer face of~$G'-B-v_i^b$, such that every vertex in $C$ has the same~$x$-coordinate as in~$\Gamma'_B-v_i^b$ (we can easily find such a drawing~$C$ by adding~$v_i$ to the convex hull of~$\Gamma'_B-v_i^b$ in a strictly convex fashion).
Then, Lemma~\ref{lem:H&N} (for vertical morphs) (potentially followed by a vertical shearing transformation in order to remove horizontal edges)  applied to~$\Gamma'_B-v_i^b$ and~$C$ yields the desired drawing~$\Gamma''$.
Analogous to Section~\ref{sec:t1-monotone}, by combining Observation~\ref{lem:uni-morph2}, Lemma~\ref{lem:uni-morph} and Lemma~\ref{lem:convexity-inc} we conclude that the vertical morph $\langle \Gamma'_B-v_i^b,\Gamma''\rangle$ is planar and convexity-increasing.

\noindent{\bf Case 4:} we are not in Case 1, Case 2 or Case 3.
We reduce to Case 2 or Case 3:
using a shearing transformation along the $x$-axis (or along the $y$-axis), we obtain a drawing~$\Gamma''_B$ of~$G'-B$ satisfying the preconditions of Case~3 (or Case~2).
Analogous to Section~\ref{sec:t1-monotone}, by combining Observation~\ref{lem:uni-morph2}, Lemma~\ref{lem:uni-morph} and Lemma~\ref{lem:convexity-inc} we conclude that the linear morph $\langle \Gamma'_B,\Gamma''_B\rangle$ is planar and convexity-increasing.
\end{proof}

Starting with the drawing~$\Gamma'$ of~$G'$ and iterating Lemma~\ref{lem:removeLabelBmorph}, we obtain a strictly convex drawing of~$G'-V^b$.
By construction, we can simply remove all the vertices of~$V^{ac}$ to obtain a strictly convex drawing of~$G$.

\begin{obs}
Let $\Gamma'_{ac}$ be a strictly convex drawing of~$G'-V^b$.
Then, $\Gamma'_{ac}-V^{ac}$ is a strictly convex drawing of~$G$.
\end{obs}

We summarize:

\begin{proof}[\proofof{Theorem~\ref{thm:morph-to-convex}}]
We analyze the three steps of the algorithm individually.

\noindent{\bf Step 1:}
We begin by augmenting~$G$ and~$\Gamma$ to~$G'$ and~$\Gamma'$.
As discussed in the corresponding section, this can be done in~$O(n^2)$ time.

\noindent{\bf Step 2:}
Next, we apply the algorithm from Section~\ref{sec:non-convex-outer} to~$G'$ and~$\Gamma'$.
This algorithm was designed for $3$-connected graphs.
However, note that Lemma~\ref{lem:outerFace3con} applies to internally $3$-connected graphs as well).
Therefore, the algorithm also works for internally $3$-connected graphs as long as adding convex hull edges and then successively removing them never creates a plane graph which is not internally $3$-connected.
This is the case by Corollary~\ref{cor:canAddCh}.
Thereby, we obtain in $O(n'~^{1+\omega /2})\subseteq O(n^{1+\omega /2})$ time a convexity-increasing morph from~$\Gamma'$ to a strictly convex drawing of~$G'$, where~$n'$ is the number of vertices of the augmented graph~$G'$.
This morph is also convexity-increasing with respect to the subdrawing of~$G$ as every internal angle of~$\Gamma$ is also internal in~$\Gamma'$. 

As discussed in the last paragraph of Section~\ref{sec:non-convex-outer}, the upper bound on the number of morphing steps guaranteed by the algorithm is $\mathrm{max}\lbrace  2,r'+1\rbrace + 2\rho' +1$, which can be bounded by~$1.5n'+2$.
Here,~$\rho'$ denotes the number of pockets of~$\Gamma'$.
In fact, the bound can be improved to~$1.5n+2<1.5n'+2$ by observing that~$r=r'$ and~$\rho =\rho'$, where~$n$ and~$\rho$ denote the number of vertices and pockets, respectively, of the original drawing~$\Gamma$.
The latter equality is obvious.
For the former equality, observe that each vertex~$v_i^a\in V^{ac}, 2\le i\le k$ has an angle of~$\pi$ and, hence, it has no reflex angle.
Further, the construction of~$\Gamma'$ ensures that an outer angle at a vertex~$v_i^b\in V^b$ is reflex if and only if the outer angle at the corresponding vertex $v_i$ is reflex in~$\Gamma$.
Other angles at~$v_i^b$ can not be reflex.
Moreover, if the outer angle at~$v_i$ is reflex in~$\Gamma$ then~$v_i$ has no reflex angle in $\Gamma'$.
Consequently, we can charge the inner reflex angles of the vertices~$v_i^b$ to their counterparts~$v_i$.
Finally, the vertices $v_0,v_1^a,v_k^c,v_k+1$ of each pocket~$P_e'$ belong to the convex hull of~$\Gamma'$ and, thus, they do no have any internal reflex angles.
The convexity status of the remaining angles is untouched and, hence,~$r=r'$ as claimed.
Altogether, we obtain the improved bound
\[
\mathrm{max}\lbrace  2,r'+1\rbrace + 2\rho' +1=\mathrm{max}\lbrace  2,r+1\rbrace + 2\rho +1\le 1.5n+2,
\]
where the last inequality was already discussed in the last paragraph of Section~\ref{sec:non-convex-outer}.

\noindent{\bf Step 3:}
Finally, we iteratively apply Lemma~\ref{lem:removeLabelBmorph} to the strictly convex drawing of~$G'$ that was obtained in the previous step.
Each application increases the number of vertices of~$G$ on the convex hull.
Thus, by induction we arrive at a strictly convex drawing of~$G$.
Each of the intermediary morph steps is convexity-increasing with respect to the respective augmented graph.
Once again, since every internal angle of~$\Gamma$ remains internal in (all) the augmented graph(s), we have that the morphing sequence is convexity-increasing for~$G$ as well.
The number of morphing steps is bounded by~$2n$ and the time required to obtain the entire sequence sums up to~$O(n'~^{1+\omega /2})\subseteq O(n^{1+\omega /2})$.

Summing up, we end up with $3.5n+2$ morphing steps and a runtime of~$O(n^{1+\omega /2})$.
 \end{proof}

%% file: Tutte.tex
\section{Using Tutte's Algorithm to Find Convex Drawings Preserving $y$-Coordinates}
\label{sec:Tutte}

In this section we prove
Lemma~\ref{lem:H&N} using Tutte's graph drawing algorithm. 
\changednew{This reduces the problem to solving a linear system. 
Applying a generalized method due to Alon and Yuster~\cite{alon2013matrix} for solving a linear system whose non-zero's in the matrix correspond to a planar graph, we obtain}
an algorithm that produces a straight-line strictly convex redrawing preserving $y$-coordinates and that runs in $O(n^{\omega/2})$ time, where $\omega$ is the matrix multiplication exponent.

In his paper, ``How to Draw a Graph,''~\cite{tutte1963} Tutte showed that any 3-connected planar graph $G=(V,E)$ with a fixed convex drawing $C$ of its outer face has  a convex drawing with outer face~$C$ that can be obtained by solving a system of linear equations.  
\changednew{For each $v \in V$ let the variables $(x_v,y_v)$ represent the coordinates of vertex $v$.
Let $V_I$ be the internal vertices of $G$ and let $V_B$ be the vertices of the outer face. 
For each vertex $v \in V_B$ let $(x_v^b, y_v^b)$ 
be its (fixed) coordinates in $C$.  Let $d_v$ be the degree of vertex $v$.  Consider the system of equations:}
\begin{align}
\forall u \in V_I \ \ \ \ \ \ (x_u, y_u)  &= \sum_{(u,v) \in E}  \frac{1}{d_u} (x_v,y_v),\nonumber\\
\forall u \in V_B \ \ \ \ \ \ (x_u, y_u)  &= (x_u^b, y_u^b).
\label{eq:Tutte-sym}
\end{align}

Tutte proved that this system of equations has a unique solution 
and that the solution gives a convex drawing of $G$ with outer face $C$.
In fact, the drawing is even strictly convex if $C$ is strictly convex.

Tutte's Theorem was originally stated for $3$-connected graphs. However, it is well known that Tutte's proof also applies to the more general class of internally $3$-connected graphs since it only uses Property \ref{I3} of Lemma~\ref{lem:int3conDefs}. For the special case of  \emph{strictly} convex drawings of the outer face, the generalization from $3$-connected to internally $3$-connected graphs is easy to prove:
Suppose graph $G$ has an external separation pair $(u,v)$.
We make use of Observation~\ref{lem:externalSepPair}.  
Vertices $u$ and $v$ lie on the outer face, and their removal separates the graph into two connected components $D$ and $D'$.  
In the strictly convex drawing $C$ of the outer face, a line segment joining $u$ to $v$ will separate $C$ into two strictly convex subpolygons since by Property~\ref{E2} of Observation~\ref{lem:externalSepPair} the vertices $u$ and $v$ are not consecutive on~$C$.
By induction, Tutte's algorithm will draw each of the two components with strictly convex faces in the appropriate subpolygon, and---in case $(u,v)$ is not an edge of the graph---the face between $D$ and $D'$ will also be strictly convex.  
It is not necessary to apply Tutte's algorithm separately to the two components---one system of equations will do.

Tutte's theorem can be generalized to more general ``barycenter''  weights other than $1/d_u$.
Assign a weight  ${w_{u,v}>0}$ to each ordered  pair $(u,v)$ with $(u,v) \in E$ such that for each $u$ it holds that
${\sum_v  w_{u,v} = 1}$.  We emphasize that $w_{u,v}$ may be different from $w_{v,u}$. 
\changednew{Consider the system of equations:
\begin{align}
\forall u \in V_I \ \ \ \ \ \ (x_u, y_u)  &= \sum_{(u,v) \in E}  w_{u,v} (x_v,y_v),\nonumber\\
\forall u \in V_B \ \ \ \ \ \ (x_u, y_u)  &= (x_u^b, y_u^b).
\label{eq:Tutte}
\end{align}
This system also has a unique solution that gives a convex drawing of $G$ with outer face $C$, and a strictly convex drawing of $G$ if $C$ is strictly convex.
This generalization was first stated by Floater in 1997~\cite{floater-triangulation} for triangulations and one year later~\cite{Floater} for general 3-connected planar graphs, though the result is not stated as a theorem in either case.  Floater proved that the constraint matrix is non-singular, and, for the rest, said that Tutte's proof \footnote{ Colin de Verdiere et al.~\cite{Colin-de-Verdiere} point out that Tutte's original proof is complicated because Tutte is also re-proving Kuratowski's theorem, and they recommend the simpler proof by Richter-Gebert~\cite{Richter-Gebert}.}
carries over.

An explicit statement that the linear system (\ref{eq:Tutte}) has a unique solution that gives a strictly convex drawing of $G$ if $C$ is strictly convex is due to 
Gortler, Gotsman, Thurston in 2006 \cite[Theorem 4.1] {Gortler-2006}. They give a new proof using ``one-forms''.}

We can now give an alternate proof of Hong and Nagamochi's result: 

\begin{proof}[\proofof{Lemma~\ref{lem:H&N}}]
We must show that there is a strictly convex drawing of $G$ with outer face~$C$ that preserves the $y$-coordinates of the vertices from drawing $\Gamma$.  
Our idea is to do this in two steps, first choosing the barycenter weights to force vertices to lie at the required $y$-coordinates, and then solving system~(\ref{eq:Tutte}) with these barycenter weights to determine the $x$-coordinates.

For the first step, we solve the following system separately for each $u \in V_I$:
\begin{align}
y_u   = \sum_{(u,v) \in E}  w_{u,v} y_v, \ \ \ \ \ 
1  = \sum_v  w_{u,v}
\label{eq:barycenter}
\end{align}
Here the $y$'s are the known values from $\Gamma$  and the $w_{u,v}$'s are variables.   There are two equations and $d_u > 2$ variables, so the system is under-determined and can easily be solved: 
Because $\Gamma$ has $y$-monotone faces, every vertex $u \in V_I$ has neighbors below and above.  Let $N_u^+$ be the neighbors of $u$ that lie above $u$ in $\Gamma$.
Let $d_u^+ = | N_u^+ |$.  Similarly, let $N_u^-$ be the neighbors of $u$ that lie below $u$ in $\Gamma$ and let $d_u^- = | N_u^- |$.  Compute the average $y$-coordinate of $u$'s neighbors above and below: 
\begin{align*}
y_u^+ = \sum_{v \in N_u^+} \frac{1}{d_u^+} y_v\ \ \ \ \ \ 
y_u^- = \sum_{v \in N_u^-} \frac{1}{d_u^-} y_v
\end{align*}
Observe that $y_u$ lies between $y_u^+$ and $y_u^-$.  Thus we can find a value $t_u$, $0 < t_u < 1$, such that
\begin{align*}
y_u & = t_u y_u^+ + (1-t_u) y_u^- \\
& = \sum_{v \in N_u^+} \frac{t_u}{d_u^+} y_v   +  \sum_{v \in N_u^-} \frac{1-t_u}{d_u^-} y_v
\end{align*} 
Therefore, setting $w_{u,v} =  \frac{t_u}{d_u^+}$ for $v \in N_u^+$ and $w_{u,v} =  \frac{1-t_u}{d_u^-}$ for $v \in N_u^-$, yields a solution to~(\ref{eq:barycenter}).
\changednew{Observe that $w_{u,v} > 0$ for all $(u,v) \in E$.}

Given values $w_{u,v} > 0$ satisfying the constraints~(\ref{eq:barycenter}) for all $u \in V_I$, we then solve equations~(\ref{eq:Tutte}) to find values for the $x_u$'s.
By Tutte's generalized result, this provides a strictly convex drawing of $G$ with outer face $C$ while preserving $y$-coordinates.

\changednew{
It remains to discuss how to obtain the claimed run-time. 
Recall that we assume a real RAM model of computation---in particular, each arithmetic operation takes unit time.
Observe that solving the system~(\ref{eq:barycenter}) to find the appropriate weights $w_{u,v}$ based on the $y$-coordinates takes linear time.  The significant aspect is solving Tutte's generalized system of equations~(\ref{eq:Tutte}).  

Tutte's method gives rise to a linear system $Ax=b$ where $A$ is a matrix with a row and column for each vertex, and where the non-zeros in the matrix correspond to edges in the planar graph. 
In more detail, Tutte's original linear system~(\ref{eq:Tutte-sym}) can be re-written as
\begin{align*}
\forall u \in V_I \ \ \ \ \ \ d_u (x_u, y_u)  &= \sum_{(u,v) \in E}  (x_v,y_v),
\end{align*}
so the non-zeros in the constraint matrix,  $A$, consist of entries $-d_u$ down the main diagonal, and $a_{u,v} = a_{v,u} = 1$ if $(u,v)$ is an edge.   
The equations for vertices in $V_B$ give an extra part of the  constraint matrix that is just an identity matrix.
The crucial property is that the constraint matrix is symmetric.  
In fact, symmetry holds more generally if we 
choose weights or ``stresses'' $s_{u,v} = s_{v,u}$ and define $w_{u,v} = s_{u,v}/ \sum_{(u,z) \in E} s_{u,z}$. 
Tutte's original theorem is the special case where $s_{u,w}=1$ for all $(u,v)$.

When the constraint matrix $A$ is symmetric with non-zero's corresponding  to the edges of a planar graph, 
the system $Ax = b$ can be solved in $O(n^{\omega/2})$ arithmetic operations using the generalized nested dissection method of Lipton, Rose, Tarjan~\cite{Lipton,DBLP:journals/siamcomp/LiptonT80}.
The fact that this applies to Tutte's original algorithm is well-known in graph drawing, 
see for example~\cite{DBLP:conf/compgeom/ChrobakGT96,Ribo-small-grid}. 
However, nested dissection does not apply when the matrix $A$ is not symmetric, so in particular, it does not apply to the linear 
system~(\ref{eq:Tutte}).\footnote{Although we mistakenly claimed this in a preliminary version~\cite{Kleist-WG-2018} of this paper.} 
For the more general case of an asymmetric matrix $A$ we need the following result of 
Alon and Yuster from 2013~\cite{alon2013matrix}.  They consider a linear system $Ax=b$ where $A$ has a row and a column for each vertex of an \emph{associated graph} $G$ and there is an edge $(u,v)$ in $G$ if and only if $a_{u,v} \ne 0$ or $a_{v,u} \ne 0$ (the diagonal entries of $A$ play no role in the definition of $G$).

\begin{theorem}[Theorem 1.1 in~\cite{alon2013matrix}, specialized to $\mathbb Q$ and to planar graphs]
Let $A \in {\mathbb Q}^{n \times n}$ be a nonsingular matrix and let $b \in {\mathbb Q}^n$. If the graph associated with $A$ is planar, then $Ax = b$ can be solved in $O(n^{\omega/2})$ time. 
\end{theorem} 

They assume that each arithmetic operation takes unit time, i.e., that algorithms are measured in terms of their algebraic complexity.
Moreover, the matrix is assumed to be given in an implicit form.
They also assumed that~$\omega>2$; otherwise the run-time is in~$O(n\log n)$.

In terms of practicality, they note that although the fastest matrix multiplication algorithms are only theoretical, using naive matrix multiplication gives a run time of $O(n^{1.5})$ and 
all ingredients of the algorithm become practically implementable, and Strassen's algorithm is sensible for larger $n$ and gives a run time of $O(n^{1.41})$.

Note that in our case the matrix $A$ is non-singular because of Tutte's generalized result.  
This completes the proof of Lemma~\ref{lem:H&N}.}
\end{proof}

%% file: spiral.tex
\section{Lower Bound on the Number of Morphing Steps}
\label{sec:spiral}
In this section, we show a linear lower bound on the number of required morphing steps.
\spiral*
Our proof of Theorem~\ref{thm:spiral} builds on the following result by Alamdari et al.~\cite{alamdari2016morph}.
\begin{theorem}[\cite{alamdari2016morph}]
\label{thm:spiral-2}
There exist two straight-line planar drawings~$\Gamma^\Delta(n')$ and~$\Gamma^-(n')$ of a path with~$n'$ vertices such that any planar morph between them which is composed of a sequence of linear morphing steps requires $\Omega (n')$ steps.
\end{theorem}

In the drawing~$\Gamma^-(n')$ the~$n'$ vertices $a_1,\dots,a_{n'}$ are placed on a horizontal line with $a_i$ to the left of $a_{i+1}$ for $1\le i<n'$, see Fig.~\ref{fig:spiral}(a).
In the drawing~$\Gamma^\Delta(n')$ the path forms a spiral, see Fig.~\ref{fig:spiral}(b).
More precisely, let~$e_i$ denote the edge~$\lbrace a_i,a_{i+1}\rbrace$.
Then for each $i$ with $i~\mathrm{mod}~3\equiv 1$, the edge~$e_i$ is horizontal and~$a_i$ is to the left of~$a_{i+1}$.
For each $i$ with $i~\mathrm{mod}~3\equiv 2$, the edge~$e_i$ is parallel to the line $y=\mathrm{tan}(2\pi /3)x$ and~$a_i$ is to the right of~$a_{i+1}$.
Finally, for each $i$ with $i~\mathrm{mod}~3\equiv 0$, the edge~$e_i$ is parallel to the line $y=\mathrm{tan}(-2\pi /3)x$ and~$a_i$ is to the right of~$a_{i+1}$.

In order to prove Theorem~\ref{thm:spiral}, we present a drawing~$\Gamma^{\underline{\Delta}}(n)$ of a cycle on~$n$ vertices which contains a subdrawing of $\Gamma^\Delta(n')$ for some $n'\in \Theta (n)$, see Fig.~\ref{fig:spiral}(c).
The existence of a convexifying planar morph for~$\Gamma^{\underline{\Delta}}(n)$  with~$o(n)$ linear morphing steps would imply the existence of a planar morph~$\Gamma^\Delta(n')$ and~$\Gamma^-(n')$ with~$o(n')$ linear morphing steps, constradicting Theorem~\ref{thm:spiral-2}.

 \begin{figure}[bht]
  \centering
  \includegraphics{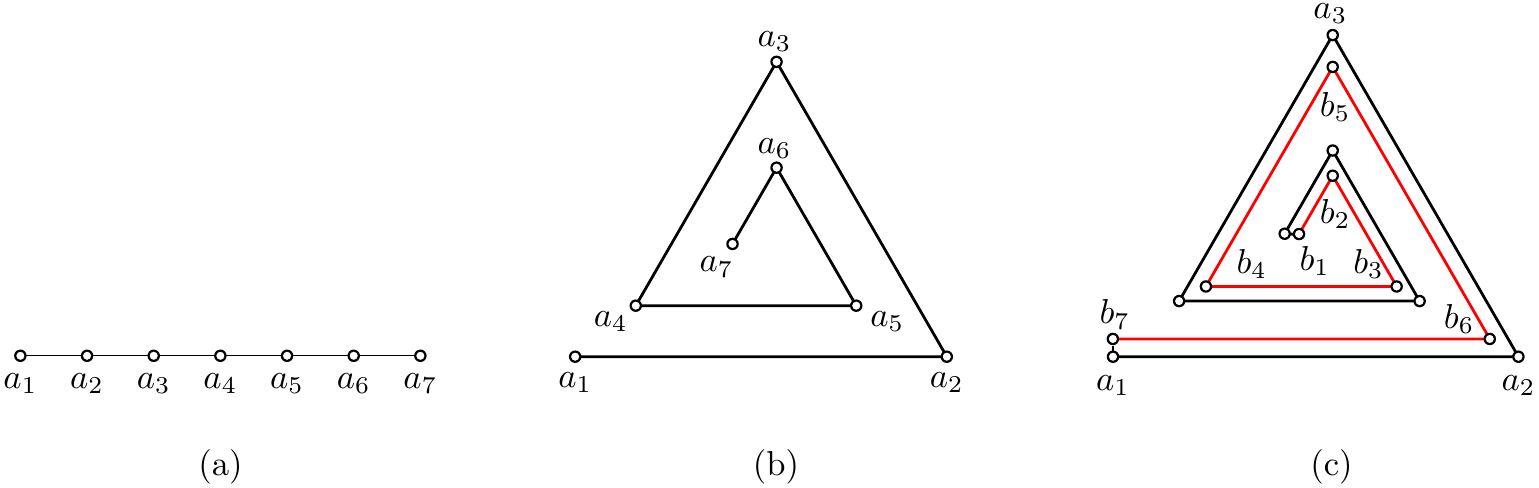}
  \caption{The drawings (a) $\Gamma^-(7)$, (b) $\Gamma^\Delta(7)$, and (c) $\Gamma^{\underline{\Delta}}(14)$.
}
  \label{fig:spiral}
 \end{figure}

\begin{proof}[\proofof{Theorem~\ref{thm:spiral}}]
Let~$n'=\lfloor n/2\rfloor$.
Let~$\Gamma^{\underline{\Delta}}(n)$ be some planar straight-line drawing of the cycle $C=(a_1,\dots,a_{n'},b_1,\dots,b_{n-n'})$ such that the induced subdrawing of the path $P=(a_1,\dots, a_{n'})$ is~$\Gamma^\Delta(n')$.
The exact realization of the path $(b_1,\dots b_{n-n'})$ is irrelevant for the purposes of this proof.
We give an example in Fig.~\ref{fig:spiral}(c).

Assume for a contradiction that there exists a morph~$\mathcal M$ composed of sequence of $o(n)$ linear morphing steps which convexifies~$\Gamma^{\underline{\Delta}}(n)$.
Restricting the morph $\mathcal M$  to the path~$P$ transforms~$\Gamma^\Delta(n')$ into a reflex chain~$\Gamma_P$ on the boundary of a strictly convex polygon.
It is easy to find~$O(1)$ additional morphing steps that transform~$\Gamma_P$ into the drawing $\Gamma^-(n')$; 
for example  we can intermediately aim for coordinates of $a_1$ and $a_n$ which are extreme in some direction 
as in the proof of Lemma~\ref{lem:outerFace3con} (in fact, our situation here is much simpler, as we are not restricted to horizontal and vertical morphs anymore).
Extending~$\mathcal M$ by these additional morphs yields planar morph with $o(n)\subseteq o(n')$ linear morphing steps that transforms $\Gamma^\Delta(n')$ into $\Gamma^-(n')$.
This is a contradiction to Theorem~\ref{thm:spiral-2}.
\end{proof}

%% file: GridSize.tex
\def\w{\ensuremath{\text{w}}_c\xspace}

\section{Lower Bound on Grid Size}
\label{sec:grid}

Every 3-connected planar graph can be drawn with convex faces on a $n \times n$ grid~\cite{DBLP:journals/order/Felsner01} or with strictly convex faces on a $O(n^2)\times O(n^2)$ grid~\cite{barany2006strictly}. It is desirable to find morphs in which the intermediate drawings lie on a polynomial-sized grid, i.e., the coordinates of each vertex have a logarithmic number of bits.
In this section, we show that this is not achievable with our current approach and, more generally, with any approach that uses Hong and Nagamochi's redrawing technique.
To do so, we design a family of drawings
to show that  a single horizontal morph to a convex drawing 
may unavoidably blow up the width of the drawing from $O(n)$ to $\Omega(n!)$.   Thus, there is no hope of restricting to a polynomial-sized grid. 

\changed{
A \emph{grid-drawing} of a planar graph $G$ is a straight-line planar drawing of $G$ in which all vertices are placed at integer coordinates. The \emph{width} of a grid-drawing $\Gamma$ is the length of a horizontal side of the smallest bounding box of~$\Gamma$. 
}

\begin{lemma}\label{lem:grid}
There exists a family of grid-drawings $(\Gamma_k)_k$  of internally 3-connected graphs~$G_k$ on~$n_k$ vertices such that the width of $\Gamma_k$ is $w(\Gamma_k)=2k$ and any grid-drawing of $G_k$ in which the $y$-coordinates match those of $\Gamma_k$ and in which every inner face is convex has at least width $\w(\Gamma_k)\geq 4^{k-1}(2k-2)!\in\Omega(w(\Gamma_k)!)=\Omega(n_k!)$.
\end{lemma}
The drawing $\Gamma_1$ of $G_1$ is depicted in Fig.~\ref{fig:blowUpGraphFamilyDef}.
The drawing $\Gamma_{k+1}$ of $G_{k+1}$ is obtained from
$\Gamma_k$ by introducing four new $y$-coordinates (called \emph{levels}) and six new vertices:
Introduce a cycle on six vertices as the new outer face and four new internal edges as shown in Fig.~\ref{fig:blowUpGraphFamilyDef}.
\begin{figure}[bht]
\centering
\includegraphics[page=1]{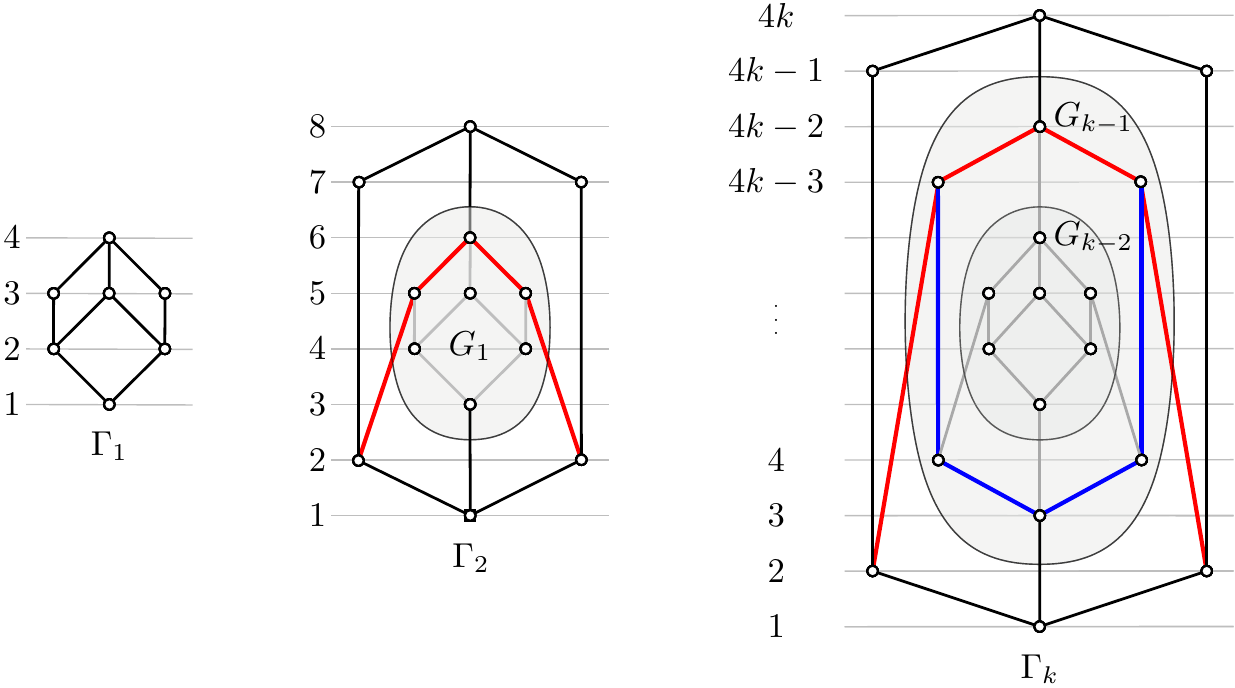}
\caption{Illustration of the drawings $\Gamma_k$ for small values of $k$. Any redrawing of $\Gamma_k$ that preserves the $y$-coordinates and has convex inner faces   needs an exponential increase in width.}
\label{fig:blowUpGraphFamilyDef}
\end{figure}

Note that $G_k$ has exactly two vertices of
degree two on the outer face. It follows directly from the construction that
$G_k$ has $6k+1$ vertices, $4k$ levels, and admits a grid-drawing of width at most $2k$.
It is easy to see that if $G_{k-1}$ has a
drawing with convex inner  faces that preserves the levels of $\Gamma_{k-1}$, then the analogous statement holds for $G_{k}$; we
can simply shift the degree-$2$ vertices in $G_{k-1}$ and the new vertices on
levels $2$ and $4k-1$ far enough outwards.

We now argue that any convex drawing necessarily blows up the width with the following geometric observation.
A similar idea was used by Lin and Eades \cite{lin2003} for the construction of hierarchical drawings where every straight-line drawing has a large width.

\begin{obs}\label{obs:increase}
Consider a grid-drawing of a path $(a_1,b_1,c,b_2,a_2)$ such that~$c$ belongs to level~$k$; $b_1$ is left of~$b_2$ on level~$k+1$; and~$a_1$ is left of~$a_2$ on level~$k+j$ where~$j\ge 2$.
Let~$D$ denote the distance of~$b_1$ and~$b_2$.
If~$a_1$ is not to the right of the oriented line~$\overrightarrow{cb_1}$ and~$a_2$ is not to the left of the oriented line~$\overrightarrow{cb_2}$, then the distance of~$a_1$ and~$a_2$ is at least~$jD$.
\end{obs}
For an illustration of Observation~\ref{obs:increase} consider Fig.~\ref{fig:obs}. We use it to prove Lemma \ref{lem:grid}.
\begin{figure}[htb]
 \centering
\includegraphics[page=3]{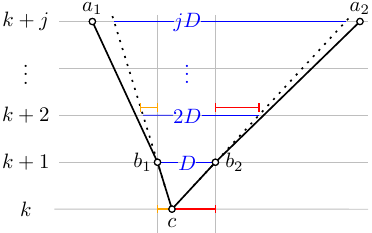}
\caption{Illustration of Observation~\ref{obs:increase}.}
\label{fig:obs}
\end{figure}
\begin{proof}[Proof of Lemma~\ref{lem:grid}]
Since the third level of~$\Gamma_1$ contains three vertices, we have $\w(\Gamma_1)\geq 2$. 
For $k\geq 2$, we prove the claim by induction with the stronger induction hypothesis that in any level-preserving grid-drawing $\Gamma_k'$ of $G_k$ with convex inner faces, the two vertices on the second level have distance at least $\w(\Gamma_k)$; and that $\w(\Gamma_k)\geq 4^{k-1}(2k-2)!$.

For the base case, consider any level-preserving redrawing $\Gamma_2'$ of $\Gamma_2$ with convex inner faces. Note that in any grid-drawing the three vertices on level 5 need a width of at least 2. 
Applying Observation~\ref{obs:increase} to the outermost vertices on levels 6,5,2 in $\Gamma_2'$ (highlighted by the red path in Fig.~\ref{fig:blowUpGraphFamilyDef}), it follows that in $\Gamma_2'$ the distance of the two vertices on level 2 is at least 8 and, thus, also $\w(\Gamma_{2})\geq 8$.

For the induction step with $k\geq 3$, consider a level-preserving redrawing $\Gamma_{k}'$ of $\Gamma_{k}$ in which all inner faces are convex. We apply  Observation~\ref{obs:increase} twice in $\Gamma_{k}'$.
By the induction hypothesis, the distance of the vertices on level 4 in $\Gamma_k'$ is at least $\w(\Gamma_{k-1})
$ since $\Gamma_k'$ contains a level-preserving redrawing of $\Gamma_{k-1}$ in which all inner faces are convex.
Applying Observation~\ref{obs:increase} to the outermost vertices on levels $3$, $4$, and $4k-3$ in $\Gamma_k'$ (highlighted by the blue path in Fig.~\ref{fig:blowUpGraphFamilyDef}) shows that the vertices on level $4k-3$ in $\Gamma_k'$ have a distance of at least $(4k-6)\cdot \w(\Gamma_{k-1})$.
Applying Observation~\ref{obs:increase} to the outermost vertices on levels $2$, $4k-3$, and $4k-2$  (highlighted by the red path in Fig.~\ref{fig:blowUpGraphFamilyDef}) shows that the vertices on level $2$  in $\Gamma_k'$ have distance at least $(4k-4)\cdot (4k-6)\w(\Gamma_{k-1})=:\delta$. Since $\w(\Gamma_{k-1})\geq 4^{k-2}(2k-4)!$, a simple calculation yields that $\delta= 4^{k-1}(2k-2)!$ and, thus, the distance of the vertices in level~2 in $\Gamma_k'$ is at least $\w(\Gamma_k)\geq \delta= 4^{k-1}(2k-2)!$. 
\end{proof}

Note that each graph $G_k$ resulting from this construction is internally 3-connected by Lemma~\ref{lem:convConn}
since it has a convex drawing. However, $G_k$ is not
3-connected. We can obtain the same result for 3-connected graphs simply by
\changed{contracting the two outer edges incident to the vertex on level~$4k$. Note that these contractions remove the two degree-$2$ vertices on the outer face of~$G_k$.}

%% file: Appendix.tex
\appendix\clearpage

\section{\changed{Previous Versions of Lemma~\ref{lem:H&N}}}
\label{appendix:H&N}

In this appendix we give more details on the 
\changed{versions of  Lemma~\ref{lem:H&N} proved by Hong and Nagamochi~\cite{hn-2012} and Angelini et al.~\cite{angelini-convex-drawings-2015}.}  
First we introduce some terminology from~\cite{hn-2012}.

A \emph{hierarchical graph} is a graph with vertices assigned to \emph{layers} which are horizontal lines.  
A \emph{level drawing} of a hierarchical graph maps each vertex to a point on its layer line, and each edge to a $y$-monotone curve.  Each edge can be directed upwards, so there is an underlying directed graph, and the standard notions of sink and source.  
A \emph{hierarchical-st plane graph} is a hierarchical graph that has a level drawing that is planar and that has only one source and one sink.  This is equivalent to our definition that all faces are $y$-monotone.

\begin{lemma} [Hong and Nagamochi~\cite{hn-2012}]
\label{lem:H&N-original}
Let $G$ be 
\changed{an internally 3-connected}  
hierarchical-st plane graph.  Then a drawing of the outer face of $G$ on a strictly convex polygon $C$ can be extended in $O(n^2)$ time to a convex drawing of $G$ by choosing an $x$-coordinate for each internal vertex.
\end{lemma}

Implicit in \changed{the above} statement is that the polygon $C$ and the convex drawing of $G$ respect the layers of the hierarchical graph.  
Hong and Nagamochi 
\changed{also gave} necessary and sufficient conditions to deal with the case where the drawing of $C$ need not be strictly convex. 

The drawings generated by Lemma~\ref{lem:H&N-original} are, in general, not strictly convex since the corresponding recursive algorithm draws certain paths of the graph by placing all its vertices on a common line.
\changed{Angelini et al.~\cite[Theorem 6]{angelini-convex-drawings-2015} describe an extension of the algorithm that perturbs the vertex coordinates in order to avoid angles of size $\pi$ and, thereby, strengthen the above lemma to obtain \emph{strictly} convex drawings.
They do not analyze the runtime of their extension.}

\section{Morphing Algorithm by Alamdari et al.}
\label{appendix:previousMorphing}

In this section we give some details to justify Theorem~\ref{thm:improved-morph}.  
The morphing algorithm of Alamdari et al.~\cite{alamdari2016morph} depended on the
basic result that it is possible to morph a triangulation so that a given
quadrilateral formed by two adjacent triangles becomes convex.  This
problem---called ``Quadrilateral Convexification''---was solved in Section 6 of the
paper by means of a single unidirectional morph achieved by applying Hong and
Nagamochi's algorithm to redraw the triangulation so that the quadrilateral is
convex.  The run-time is $O(n^2)$ because of Hong and Nagamochi's algorithm.  This
is improved to $O(n^{\omega/2})$ by our Lemma~\ref{lem:H&N}.  

In fact, Quadrilateral Convexification, applied $O(n)$ times,  was the bottleneck
that gave a run-time of $O(n^3)$ for the algorithm of Alamdari et
al.~\cite{alamdari2016morph}, as we now justify briefly.
Their algorithm consisted of a preliminary step to find a compatible triangulation,
thereby reducing the general problem to the case of triangulations.  
Then the main part of the algorithm iterated a contraction step that reduced the
number of vertices by 1 in each step.  Finally, there was a ``clean-up'' step that
replaced contraction by ``almost-contraction.''  
The preliminary compatible triangulation step (Theorem 4.1) involved $O(n)$ calls to
Quadrilateral Convexification.  As noted at the end of the proof of Theorem 4.1,
this dominated the other work.  Thus our  Lemma~\ref{lem:H&N} improves the run-time
of the preliminary step from $O(n^3)$ to $O(n^{1 + \omega/2})$.
The main iterative algorithm (Theorem 5.1) consisted of $O(n)$ iterations, and each
iteration involved a constant number of applications of Quadrilateral
Convexification.  Thus the original run-time of $O(n^3)$ is reduced to  $O(n^{1 +
\omega/2})$.
Finally, the clean-up step (Theorem 3.3) did not use Quadrilateral Convexification
and ran in time $O(n)$ per iteration, for a total of $O(n^2)$.